\documentclass[a4paper,11pt,fleqn]{article}
\usepackage[margin=1.25in]{geometry}
\pdfoutput=1

\usepackage{natbib}
\usepackage{fullpage}
\usepackage{amssymb,amsfonts,amsmath,amsthm,bm}
\usepackage{mathrsfs}
\usepackage{amscd}
\usepackage{xcolor}
\usepackage{array}
\usepackage{caption}
\usepackage{subcaption}
\usepackage{enumerate}
\usepackage{doi}
\usepackage{tikz}
 \usetikzlibrary{positioning}
\usetikzlibrary{arrows,automata}
\usetikzlibrary{shapes.multipart}
\usepackage[linesnumbered,ruled]{algorithm2e}
\usepackage{hyperref}
\hypersetup{
    breaklinks=true,
    colorlinks=false,
    citebordercolor=teal,
    filebordercolor=purple,
    linkbordercolor=violet,
    urlbordercolor=violet  % splits links across lines % etc.
        }

\title{Optimal Approximate Minimization of One-Letter Weighted Finite Automata}
\usepackage{authblk}
\date{}

\author[1,2]{Clara Lacroce\footnote{Corresponding author: clara.lacroce@mail.mcgill.ca  \href{https://orcid.org/0000-0002-2087-7630}{ORCID}\\Preprint under consideration for publication in \emph{Mathematical Structures in Computer Science}.}}
\author[3]{Borja Balle}
\author[1,2]{Prakash Panangaden}
\author[2,4,5]{Guillaume Rabusseau}
\affil[1]{School of Computer Science, McGill University, Montr\'eal, Canada}
\affil[2]{Mila, Montr\'eal, Canada}
\affil[3]{DeepMind, London, United Kingdom.}
\affil[4]{DIRO, Universit\'e de Montr\'eal, Montr\'eal, Canada}
\affil[5]{CIFAR AI Chair}

\newtheorem{therm}{Theorem}

\newtheorem{corollary}[therm]{Corollary}
\newtheorem{definition}[therm]{Definition}
\newtheorem{lemma}[therm]{Lemma}
 \newtheorem{example}{Example}

\newcommand{\mat}[1]{\mathbf{#1}}
\newcommand{\norm}[1]{\|#1\|}
\newcommand{\rank}{\operatorname{rank}}

\newcommand{\C}{\mathbb{C}}
\newcommand{\N}{\mathbb{N}}
\newcommand{\Z}{\mathbb{Z}}
\newcommand{\R}{\mathbb{R}}
\newcommand{\A}{\mat{A}}
\renewcommand{\H}{\mat{H}}

\newcommand{\mP}{\mat{P}}
\newcommand{\mQ}{\mat{Q}}
\newcommand{\mT}{\mat{T}}

\newcommand{\balpha}{\boldsymbol{\alpha}}
\newcommand{\bbeta}{\boldsymbol{\beta}}
\newcommand{\wfa}{\langle \balpha , \{\A_a\},  \bbeta \rangle}
\newcommand{\wa}{\langle \balpha , \A,  \bbeta \rangle}

\begin{document}

\maketitle

\begin{abstract}

    In this paper, we study the approximate minimization problem of weighted finite automata (WFAs): to compute the best possible approximation of a WFA given a bound on the number of states. By reformulating the problem in terms of Hankel matrices, we leverage classical results on the approximation of Hankel operators, namely the celebrated Adamyan-Arov-Krein (AAK) theory. 
    
    We solve the optimal spectral-norm approximate minimization problem for irredundant WFAs with real weights, defined over a one-letter alphabet. We present a theoretical analysis based on AAK theory,  and bounds on the quality of the approximation in the spectral norm and $\ell^2$ norm. Moreover, we provide a closed-form solution, and an algorithm, to compute the optimal approximation of a given size in polynomial time. 
\end{abstract}
\section{Introduction}

%%%%%% COMMENTS FROM INTERNAL DEEPMIND REVIEW, TO BE ADDRESSED:

% * page 1-2, on the topic of "quantifying the error". At the end of the first paragraph, you mention that when used in learning, the two steps are presented in the context of training data + test data. It might be good to perhaps keep the theme in the second paragraph and discuss data-driven approaches to quantify the error introduced by some choice of minimal WFA. In a way, the end of paragraph 1 implies that using a smaller (approximate) WFA is a form of normalization to reduce data overfitting; there are some known data driven approaches to compare different normalization techniques. Might be good to mention where these fail and how this work extends on existing analytical approaches. If you have thought on ways of combining this with data-driven approaches, discussing this in the intro or the conclusion would be even better :) 

    Weighted finite automata (WFAs) are an expressive class of models representing functions defined over sequences. The \emph{approximate minimization problem} is concerned with finding an automaton that approximates the behaviour of a given minimal WFA, while being smaller in size. Clearly, the two automata compute different languages, so the objective is to minimize the approximation error~\citep{Balle15,Balle19}. Approximate minimization can be particularly useful in the context of spectral learning algorithms~\citep{BaillySpectral,Hsu,BalleCLQ14,ballewill}. When applied to a learning task, such algorithms can be viewed as working in two steps. First, they compute a minimal WFA that explains the training data exactly. Then, they obtain a model that generalizes to the unseen data by producing a smaller approximation to the minimal WFA, thus preventing overfitting of the data.

    A key point in solving approximation tasks is to choose how to quantify the error. We propose to rewrite the problem in terms of Hankel matrices, mathematical objects strictly related to WFAs, and to measure the error in terms of the \emph{spectral norm}. This allows us to exploit the work of Adamyan, Arov and Krein which has come to be known as AAK theory~\citep{AAK71}: a series of results connecting the theory of complex functions to Hankel matrices. The core of this theory provides us with theoretical guarantees for the exact computation of the spectral norm of the error, and a method to construct the optimal approximation. We show that the spectral norm of the Hankel matrix of a WFA can be computed accurately in polynomial time (cubic in the number of states of the automaton). This is a great advantage compared, for example, to behavioural norms, which are easier to interpret but harder to compute~\citep{Balle17,pascale_journal}. The spectral norm has another advantage over WFA-specific behavioural metrics. In fact, an important extension of this work is the application of the method to other classes of models. In the one-letter case, a similar algorithm can be found to approximate a black-box model over sequential data using a WFA \citep{AAK-RNN}. With this in mind, we think that it is preferable to consider a norm defined on the input-output function -- or the Hankel matrix --  rather than the parameters of the specific model considered.
    
    We summarize our main contributions:
    \begin{itemize}
        \item We apply AAK theory to the approximate minimization problem of WFAs by establishing a correspondence between the parameters of a WFA and the coefficients of a complex function on the unit circle. To the best of our knowledge, this paper represents the first attempt to apply AAK theory to WFAs. 
        \item We present a theoretical analysis of the optimal spectral-norm approximate minimization problem of WFAs, based on their connection with finite-rank infinite Hankel matrices. We provide a closed form solution for real weighted automata $A=\wa$ over a one-letter alphabet, under the assumption $\rho(\mat{A})<1$ on the spectral radius. We bound the approximation error, both in terms of the Hankel matrix (spectral norm) and of the rational function computed by the WFA ($\ell^2$ norm).
        \item We propose a self-contained algorithm that returns the unique optimal spectral-norm approximation of a given size in polynomial time.
        \item We tighten the connection, made in~\cite{Balle19}, between WFAs and discrete dynamical systems, by adapting some of the control theory concepts to this setting, \emph{e.g.} the \emph{allpass system} \citep{Glover}.
    \end{itemize}

    In this paper we present and expand the results of our previous work \citep{AAK-WFA}. The contents of this paper are organized as follows. In Section \ref{sec:background} we define the notation that will be used throughout the paper and review a series of well-known results from the theory of automata and from functional analysis. In Section \ref{sec:aakapproxmin} we establish the framework to reformulate the approximate minimization problem in terms of Hankel operators and AAK theory. Section \ref{solution} presents the theoretical foundation of our contribution and a closed-form solution for our problem. Section \ref{sec:algorithm} shows how to implement the algorithm derived from the solution obtained in the previous section, while in Section \ref{sec:example} we provide an example and compute the optimal approximation of a given WFA. Section \ref{relatedwork} discusses the related work in approximate minimization and control theory. Finally, in Sections \ref{future_work} and \ref{conclusion} we highlight possible directions for future work, analyze the limitations of this approach and summarize our contribution.
        
\section{Background}\label{sec:background}

    In this section, we recall the fundamental definitions and preliminary results that are used throughout the paper. After defining weighted finite automata and Hankel matrices, we will provide an overview of AAK theory. We will see in the next section that our objective is to rewrite the approximate minimization problems as low-rank approximation of a Hankel matrix. In the paper, We use AAK theory to solve the low-rank approximation problem while preserving the Hankel property.

\subsection{Preliminaries}

    We denote with $\N$, $\Z$ and $\R$ the set of natural, integers and real numbers, respectively. We use bold letters for vectors and matrices; all vectors considered are column vectors. We denote with $\mat{1}$ the identity matrix, specifying its dimension only when not clear from the context. We refer to the $i$-th row and the $j$-th column of $\mat{M}$ by $\mat{M}(i,:)$ and $\mat{M}(:,j)$. Given a matrix $\mat{M}\in \R^{p\times q} $ of rank $n$, a \emph{rank factorization} is a factorization $\mat{M}=\mP\mQ$, where $\mP \in \R^{p\times n}$, $\mQ \in \R^{n\times q}$ and $\rank(\mat{M})=\rank(\mP)=\rank(\mQ)=n$. Let $\mat{M} \in \R^{p \times q}$ of rank $n$, the compact \emph{singular value decomposition} SVD of $\mat{M}$ is the factorization $\mat{M}=\mat{U}\mat{D}\mat{V}^{\top}$, where $\mat{U}\in \R^{p\times n}$, $\mat{D}\in \R^{n\times n}$, $\mat{V}\in \R^{q \times n}$ are such that $\mat{U}^{\top}\mat{U}=\mat{V}^{\top}\mat{V}=\mat{1}$, and $\mat{D}$ is a diagonal matrix.
    %with entries $\sigma_0 \geq \dots \geq \sigma_n > 0$. If the previous inequalities are strict, the SVD is unique. 
    The columns of $\mat{U}$ and $\mat{V}$ are called left and right \emph{singular vectors}, while the entries of $\mat{D}$ are the \emph{singular values}. The \emph{Moore-Penrose pseudo-inverse} $\mat{M}^+$ of $\mat{M}$ is the unique matrix such that $\mat{M}\mat{M}^+\mat{M}=\mat{M}$, $\mat{M}^+\mat{M}\mat{M}^+=\mat{M}^+$, with $\mat{M}^+\mat{M}$ and $\mat{M}\mat{M}^+$ Hermitian~\citep{Zhu}.
    %We remark that the pseudo-inverse can be computed using the SVD of $\mat{M}$: $\mat{M}^+=\mat{V}\mat{D}^{-1}\mat{U}^{\top}$.
    The \emph{spectral radius} $\rho(\mat{M})$ of a matrix $\mat{M}$ is the largest modulus among its eigenvalues. 

    A \emph{Hilbert space} is a complete normed vector space where the norm arises from an inner product. A linear operator  $T: X \rightarrow Y$ between Hilbert spaces is \emph{bounded} if it has finite operator norm, \emph{i.e.} $\norm{T}_{op} = \sup_{\norm{g}_X\leq 1}\norm{Tg}_Y < \infty$. We denote by $\mat{T}$ the (infinite) matrix associated with $T$ by some (canonical) orthonormal basis on $H$. An operator is \emph{compact} if the image of the unit ball in $X$ is relatively compact. Given Hilbert spaces $X, Y$ and a compact operator $T:X \rightarrow Y$, we denote its adjoint by $T^*$. The \emph{singular numbers} $\{\sigma_n\}_{n \geq 0}$ of $T$ are the square roots of the eigenvalues of the self-adjoint operator $T^* T$, arranged in decreasing order.
    A $\sigma$-\emph{Schmidt pair} $\{\boldsymbol{\xi}, \boldsymbol{\eta}\}$ for $T$ is a couple of norm $1$ vectors such that: $\mat{T}\boldsymbol{\xi}=\sigma \boldsymbol{\eta}$ and $\mat{T}^*\boldsymbol{\eta}= \sigma\boldsymbol{\xi}$.
    The Hilbert-Schmidt decomposition provides a generalization of the compact SVD for the infinite matrix of a compact operator $T$ using singular numbers and orthonormal Schmidt pairs: $\mT\mat{x}=\sum_{n\geq0}\sigma_n\langle\mat{x},\boldsymbol{\xi}_n \rangle \boldsymbol{\eta}_k$~\citep{Zhu}. The \emph{spectral norm} $\norm{\mat{T}}$ of the matrix representing the operator $T$ is the largest singular number of $T$. Note that the spectral norm of $\mat{T}$ corresponds to the operator norm of $T$.

    Let $\ell^2$ be the Hilbert space of square-summable sequences over $\Sigma^*$, with norm $\norm{f}_2^2=\sum_{x\in\Sigma^*}|f(x)|^2$ and inner product $\langle f, g \rangle = \sum_{x \in \Sigma^*}f(x)g(x)$ for $f,g \in \R^{\Sigma^*}$. Let $\mathbb{T}=\{z\in \C: |z|=1\}$ be the complex unit circle, $\mathbb{D}=\{z\in \C: |z|<1\}$ the (open) complex unit disc. Let $1<p< \infty$, $\mathcal{L}^p(\mathbb{T})$ be the space of measurable functions on $\mathbb{T}$ for which the $p$-th power of the absolute value is Lebesgue integrable. For $p=\infty$, we denote with $\mathcal{L}^{\infty}(\mathbb{T})$ the space of measurable functions that are bounded, with norm $\norm{f}_{\infty}=\sup\{|f(x)|: x\in\mathbb{T}\}$.

\subsection{Hankel matrix and Weighted Automata}
    
    Let $\Sigma$ be a fixed finite alphabet and $\Sigma^*$ be the set of all finite strings with symbols in $\Sigma$. We denote with $\varepsilon$ the empty string. Given $p,s \in \Sigma^*$, we denote with $ps$ the string obtained by their concatenation. Let $f : \Sigma^* \to \R$ be a function defined on sequences, we consider the bi-infinite matrix $\H_f \in \R^{\Sigma^* \times \Sigma^*}$ having rows and columns indexed by strings and defined by $\H_f(p,s) = f(ps)$ for $p, s \in \Sigma^*$.
    
    \begin{definition}
         A matrix $\H \in \R^{\Sigma^* \times \Sigma^*}$ is \textbf{Hankel} if for all $p, p', s, s' \in \Sigma^*$ such that $p s = p' s'$, we have $\H(p,s) = \H(p',s')$. 
    \end{definition}
    Given a Hankel matrix $\H \in \R^{\Sigma^* \times \Sigma^*}$, there is a unique function $f : \Sigma^* \to \R$ such that $\H_f = \H$. Intuitively, the Hankel property tells us that each entry of the matrix only depends on the composition of the coordinates. Since rows and columns are indexed using strings, then the value stored in each entry only depends on the string obtained by concatenating the coordinates. 
    
    Weighted finite automata are a class of models defined over sequential data.
    A \emph{weighted finite automaton} (WFA) of $n$ states over $\Sigma$ is a tuple $A = \wfa$, where $\balpha,$ $\bbeta \in \R^n$ are the vector of initial and final weights, respectively, and $\A_a \in \R^{n \times n}$ is the matrix containing the transition weights associated with each symbol $a\in\Sigma$. Every WFA $A$ with real weights realizes (or computes) a function $f_A : \Sigma^* \to \R$, \emph{i.e.} given a string  $x = x_1 \cdots x_t \in \Sigma^*$, it returns $f_A(x) = \balpha ^\top \A_{x_1} \cdots \A_{x_t} \bbeta = \balpha ^\top \A_x \bbeta$. A function $f : \Sigma^* \to \R$ is called \emph{rational} if there exists a WFA $A$ that realizes it. 
    %The \emph{rank} of the function is the size of the smallest WFA realizing $f$. 
    Given $f : \Sigma^* \to \R$, we can use the Hankel matrix $\H_f \in \R^{\Sigma^* \times \Sigma^*}$ to recover information about the weighted automaton computing $f$.
  
    \begin{therm}[\cite{CP71,Fli}]\label{fliess}
        A function $f:\Sigma^*\rightarrow \R$ is realized by a WFA $A$ if and only if $\H_f$ has finite rank. In that case, the rank of $\H_f$ corresponds to the minimal number of states of any automaton realizing $f$.
    \end{therm}
    
    Given a WFA $A = \wfa $, the \emph{forward matrix} of $A$ is the infinite matrix $\mat{F}_A \in \R^{\Sigma^* \times n}$ given by $\mat{F}_A(p,:) = \balpha ^\top \A_p$ for any $p \in \Sigma^*$, while the \emph{backward matrix} of $A$ is $\mat{B}_A \in \R^{\Sigma^* \times n}$, given by $\mat{B}_A(s,:) = (\A_s \bbeta)^\top$ for any $s \in \Sigma^*$. Let $\H_f$ be the Hankel matrix of $f$, its forward-backward (FB) factorization is: $\H_f = \mat{F} \mat{B}^\top$. A WFA with $n$ states is \emph{reachable} if $\rank(\mat{F}_A)=n$, while it is \emph{observable} if $\rank(\mat{B}_A)=n$. A WFA is \emph{minimal} if it is reachable and observable. If $A$ is minimal, the (unique) FB factorization is a rank factorization~\citep{BalleCLQ14}.  

    We recall the definition of the singular value automaton, a canonical form for WFAs \citep{Balle15}.
    \begin{definition}
	    Let $f:\Sigma^*\rightarrow \R$ be a rational function and suppose $\H_f$ admits an SVD, $\H_f = \mat{U} \mat{D} \mat{V}^{\top}$. A \textbf{singular value automaton} (SVA) for $f$ is the minimal WFA $A$ realizing $f$ such that $\mat{F}_A=\mat{U} \mat{D}^{1/2}$ and $\mat{B}_A=\mat{V}\mat{D}^{1/2}$.
    \end{definition}
    The SVA can be computed with an efficient algorithm relying on the following matrices~\citep{Balle19}.

    \begin{definition}\label{def:gramians}
        Let $f:\Sigma^*\rightarrow \R$ be a rational function, $\H_f = \mat{F} \mat{B}^\top$ the FB factorization. If the matrices $\mP=\mat{F}^\top \mat{F}$ and $\mQ=\mat{B}^\top \mat{B}$ are well defined (\emph{i.e.} the inner products of their columns are finite for any column), we call $\mP$ the \textbf{reachability Gramian} and $\mQ$ the \textbf{observability Gramian}.
    \end{definition}
    
    If $A$ is in its SVA form, the Gramians associated with its FB factorization satisfy $\mP_A = \mQ_A = \mat{D}$, where $\mat{D}$ is the matrix of singular values of the corresponding Hankel matrix.
    The Gramians can alternatively be characterized and computed~\citep{Balle19}) using fixed point equations, corresponding to Lyapunov equations when $|\Sigma|=1$~\citep{lyapunov}.
    
    \begin{therm}
        Let $|\Sigma|=1$, $A= \langle \balpha, \A, \bbeta\rangle$ a WFA with $n$ states and well-defined Gramians $\mP$, $\mQ$. Then $X=\mP$ and $Y=\mQ$ solve the equations $X-\A X\A^{\top}=\bbeta\bbeta^{\top}$ and $Y-\A^{\top}Y\A=\balpha\balpha^{\top}$.
    \end{therm}
    
    Finally, we recall the definition of \emph{generative probabilistic automata} (GPA). A WFA $A= \wfa$ is a GPA if $f_A(x)\geq 0$ for every $x$ and $\sum_{x\in\Sigma^*}f_A(x)=1$, \emph{i.e.} if $f_A$ computes a probability distribution over $\Sigma^*$.
    In general this class of automata can contain pathological examples having states not connected to any final state. To avoid these cases, we introduce the following property on the spectral radius of the transition matrix.
    
    \begin{definition}
        Given a WFA $A= \langle \balpha , \{\A_a\}_{a \in \Sigma}, \bbeta \rangle$, let $\A=\sum_{a\in\Sigma}\A_{a}$. The WFA A is \textbf{irredundant} if $\rho(\A)<1$.
    \end{definition}

\subsection{AAK Theory}\label{aak-section}

    In this section, we introduce the theory of optimal approximation for Hankel operators and complex functions known as AAK theory~\citep{AAK71}. A comprehensive presentation of the concepts recalled in this section can be found in~\citet{Nikolski,Peller}.
    
    We consider the space of square-integrable complex functions on the unit circle $\mathcal{L}^2(\mathbb{T})$. To avoid any confusion with functions defined over sequences, when dealing with complex function we make explicit the dependence on the complex variable $z=e^{it}$. Note that a function $\phi(z) \in \mathcal{L}^2(\mathbb{T})$ can be represented, using the orthonormal basis $\{z^n\}_{n \in \Z}$, by means of its Fourier series: $\phi(z)=\sum_{n \in \Z}\widehat{\phi}(n)z^n$, with Fourier coefficients $\widehat{\phi}(n)= \int_{\mathbb{T}}\phi(z) \bar{z}^n dz, \, n \in \Z$. 
    %This establishes an isomorphism between the function $\phi(z)$ and the sequence of the corresponding Fourier coefficients $\widehat{\phi}$. 
    Thus, we can partition the function space $\mathcal{L}^2(\mathbb{T})$ into two subspaces.
    \begin{definition}
        The \textbf{Hardy space} $\mathcal{H}^2$ and the \textbf{negative Hardy space} $\mathcal{H}^2_-$ on $\mathbb{T}$ are the subspaces of $\mathcal{L}^2(\mathbb{T})$ defined as:
        \begin{equation*}
            \mathcal{H}^2= \{ \phi(z) \in \mathcal{L}^2(\mathbb{T}) : \widehat{\phi}(n)=0, n < 0\},
            %\qquad \mathcal{H}^2_-=\{ \phi(z) \in  \mathcal{L}^2(\mathbb{T}) : \widehat{\phi}(n)=0, n \geq 0\}.
        \end{equation*}
        \begin{equation*}
            \mathcal{H}^2_-=\{ \phi(z) \in  \mathcal{L}^2(\mathbb{T}) : \widehat{\phi}(n)=0, n \geq 0\}.
        \end{equation*}
    \end{definition}
    Interestingly, the elements of the Hardy space can be canonically identified with the set of functions analytic in the unit disc $\mathbb{D}$, with the property that the square of their absolute value is integrable on $\mathbb{T}$ (a proof can be found in~\cite{Nikolski}). Thus, we will make no difference between these functions in the unit disc and their boundary value on the circle. Moreover, we remark that the definition of Hardy space can be generalized for any $p$-th power of the functions' absolute value, for $0<p\leq\infty$.
    
    % We can now embed the sequence space $\ell^2$ into $\ell^2(\Z)$ by ``duplicating'' each vector, \emph{i.e.} by associating $\boldsymbol{\mu}=(\mu_0, \mu_1, \dots)\in\ell^2$ to $\boldsymbol{\mu}^{(2)}=(\dots, \mu_1,\mu_0, \mu_1, \dots)\in\ell^2(\Z)$. Then, we can use the Fourier isomorphism to map the vector $\boldsymbol{\mu}^{(2)}\in\ell^2(\Z)$ to the function space $\mathcal{L}^2(\mathbb{T})$. In this way each vector $\boldsymbol{\mu}\in \ell^2$ corresponds to two functions in the Hardy spaces: 
    % \begin{equation}\label{eq:hardynotation}
    %     \mu^-(z)=\sum_{j=0}^{\infty}\boldsymbol{\mu}_j z^{-j-1} \in \mathcal{H}^2_-, \quad\quad \mu^+(z)=\sum_{j=0}^{\infty} \boldsymbol{\mu}_j z^{j} \in \mathcal{H}^2.
    % \end{equation}
     
    We define Hankel operators in the Hardy spaces.
    \begin{definition}\label{Hankel2}
        Let $\phi(z)$ be a function in the space $\mathcal{L}^2(\mathbb{T})$. A \textbf{Hankel operator} is an operator $H_{\phi}:\mathcal{H}^2 \rightarrow \mathcal{H}^2_-$ defined by $ H_{\phi}f(z)=\mathbb{P}_-\phi f(z)$, where $\mathbb{P}_-$ is the orthogonal projection from $ \mathcal{L}^2(\mathbb{T})$ onto $\mathcal{H}^2_- $ . The function $\phi(z)$ is called a \textbf{symbol} of the Hankel operator $H_{\phi}$.
    \end{definition}
    The matrix $\H_{\phi}$ associated to the Hankel operator $H_{\phi}:\mathcal{H}^2 \rightarrow \mathcal{H}^2_-$ is:
    \begin{equation}
        \H_{\phi}=\begin{pmatrix} \widehat{\phi}(-1) & \widehat{\phi}(-2) & \widehat{\phi}(-3) & \dots \\
                               \widehat{\phi}(-2) & \widehat{\phi}(-3) & \widehat{\phi}(-4) &\dots \\
                               \widehat{\phi}(-3) & \widehat{\phi}(-4) & \widehat{\phi}(-5) &\dots \\
                                \vdots& \vdots &\vdots&\ddots
            \end{pmatrix}.
    \end{equation}
    Note that this matrix satisfies the Hankel property, as each entry only depends on the composition of the corresponding coordinates.
    
    Nehari's theorem~\citep{Nehari}, characterizes bounded Hankel operators and their norm. 
    
    \begin{therm}[\cite{Nehari}]\label{thm:nehari}
        Let $\phi \in  \mathcal{L}^2(\mathbb{T})$ be a symbol of the Hankel operator on Hardy spaces $H_{\phi}:\mathcal{H}^2 \rightarrow \mathcal{H}^2_-$. Then, $H_{\phi}$ is bounded on $\mathcal{H}^2$ if and only if there exists $\psi \in \mathcal{L}^{\infty}(\mathbb{T})$ such that $\widehat{\psi}(m)=\widehat{\phi}(m)$ for all $m<0$. If the conditions above are satisfied, then:
        \begin{equation}\label{eq:nehari}
            \norm{H_{\phi}}=\inf\{\norm{\psi}_{\infty}:\widehat{\psi}(m)=\widehat{\phi}(m), \, m<0\}.
        \end{equation}
        % or equivalently:
        % \begin{equation}\label{eq:nehari2}
        %     \norm{H_{\phi}}=\inf_{f(z)\in \mathcal{H}^{\infty}}\norm{\phi(z) - f(z)}_{\infty}.
        % \end{equation}
    \end{therm}
    
    As a consequence of Theorem~\ref{thm:nehari}, if $H_{\phi}$ is a bounded operator, we can consider without loss of generality $\phi(z) \in \mathcal{L}^{\infty}(\mathbb{T})$. We remark that a Hankel operator has infinitely many different symbols, since $H_{\phi}=H_{\phi+\psi}$ for $\psi(z) \in \mathcal{H}^{\infty}$.

    \begin{definition}\label{def:rational}
        The complex function $\phi(z)$ is \textbf{rational} if $\phi(z)=p(z)/q(z)$, with $p(z)$ and $q(z)$ polynomials. The rank of $\phi(z)$ is the maximum between the degrees of $p(z)$ and $q(z)$. A rational function is \textbf{strictly proper} if the degree of $p(z)$ is strictly smaller than that of $q(z)$.
    \end{definition}
   The following result of Kronecker relates finite-rank infinite Hankel matrices to rational functions.

    \begin{therm}[\cite{kronecker}]\label{theorem:Kronecker}
        Let $H_{\phi}$ be a bounded Hankel operator with matrix $\H$. Then $\H$ has finite rank if and only if $\mathbb{P}_-\phi$ is a strictly proper rational function. Moreover the rank of $\H$ is equal to the number of poles (with multiplicities) of $\mathbb{P}_-\phi$ inside the unit disc.
    \end{therm}

%\subsubsection{AAK Theorem}\label{sec:AAKtheorem}

    We are ready to state the main result of~\citet{AAK71}. The theorem shows that for infinite dimensional Hankel matrices the constraint of preserving the Hankel property does not affect the achievable approximation error.
  
    \begin{therm}[\cite{AAK71}]\label{theorem:aakop}
        Let $H_{\phi}$ be a compact Hankel operator of rank $n$, matrix $\H$ and singular numbers $\sigma_0 \geq \dots \geq \sigma_{n-1}>0$. Then there exists a unique Hankel operator $H_g$ with matrix $\mat{G}$ of rank $k<n$ such that:
        \begin{equation}\label{eqoper}
            \norm{H_{\phi} - H_g} = \norm{\H-\mat{G}}= \sigma_k.
        \end{equation}
    \end{therm}

    We denote with $\mathcal{R}_k\subset \mathcal{H}^{\infty}_-$ the set of strictly proper rational functions of rank $k$, and we consider the set of functions:
    \begin{equation}
        \mathcal{H}^{\infty}_k=\{\psi \in \mathcal{L}^{\infty}(\mathbb{T}): \,\,\exists g \in \mathcal{R}_k, \exists l \in \mathcal{H}^{\infty},\,\, \psi=g+l \}.
    \end{equation}
    The proof of the AAK theorem is directly connected with the problem of approximating a bounded function defined on the unit circle. In fact, the theorem can be reformulated in terms of the symbols associated with the Hankel operators.
    \begin{therm}[\cite{AAK71}]\label{theorem:aaksymb}
        Let $\phi \in \mathcal{L}^{\infty}(\mathbb{T})$. Then there exists a complex function $\psi \in \mathcal{H}^{\infty}_k$ such that:
        \begin{equation}\label{eqsymbol}
            \norm{\phi - \psi}_{\infty}= \sigma_k(H_{\phi}).
        \end{equation}
    \end{therm}
    This theorem provides us with an alternative interpretation of singular numbers, relating them to the ``smoothness'' of the corresponding operator (or symbol).
    The advantage of this second formulation is that its proof is constructive, and tells us how to find the function $\psi$. 
    % \begin{proof}
    %     Let $H_{\phi}$ be a Hankel operator with symbol $\phi(z) \in  \mathcal{L}^{\infty}(\mathbb{T})$ and matrix $\H$. 
        
    %     Let $\psi(z)=g(z)+l(z)\in \mathcal{H}^{\infty}_k$ be the solution of Equation~\ref{eqsymbol}. We have:
    %     \begin{align}
    %         \norm{H_{\phi}-H_{\psi}} &= \norm{H_{\phi-{\psi}}}\\
    %                             &=\norm{H_{\sigma_k\eta^{-}_k(z)/\xi^{+}_k(z)}} \\
    %                             &\leq \sigma_k \norm{\eta^{-}_k(z)/\xi^{+}_k(z)}_{\infty}=\sigma_k
    %     \end{align}
    %     where first we used Corollary~\ref{corollary:unimodular} and then Nehari's Theorem.
    %     Now, using the definition of Hankel operator, we have:
    %     \begin{equation}
    %         \norm{H_{\phi}-H_{\psi}}=\norm{H_{\phi}-H_g}= \norm{\H-\mat{G}}\leq \sigma_k.
    %     \end{equation}
    %     Since $\norm{\H-\mat{G}} \geq \sigma_k$ (from Eckart-Young theorem~\cite{Eckart}), it follows that $\norm{\H-\mat{G}}=\sigma_k$. Note that $\mat{G}$ has rank $k$, as required, because $g\in\mathcal{R}_k$(Theorem~\ref{theorem:Kronecker}).
    % \end{proof}
    We state as a corollary the critical steps of the proof, that allows us to find the best approximating symbol.
    \begin{corollary}\label{corollary:unimodular}
        Let $\phi$ and $\{\boldsymbol{\xi}_k, \boldsymbol{\eta}_k\}$ be a symbol and a $\sigma_k$-Schmidt pair for $H_{\phi}$. A function $\psi \in \mathcal{L}^{\infty}({\mathbb{T}})$ is the best AAK approximation according to Theorem~\ref{theorem:aaksymb}, if and only if:
        \begin{equation}\label{eq:unimodular}
            (\phi-\psi)\xi^{+}_k=\sigma_k\eta^{-}_k.
        \end{equation}
        Moreover, the function $\psi$ does not depend on the particular choice of the pair $\{\boldsymbol{\xi}_k, \boldsymbol{\eta}_k\}$.
    \end{corollary}
    
    Note that the solutions of Theorem~\ref{theorem:aakop} and~\ref{theorem:aaksymb} are strictly related.
    \begin{corollary}\label{corollary:stable}
        Let $\psi \in \mathcal{H}^{\infty}_k$, with $\psi=l+g$, $g \in \mathcal{R}_k, \, l \in \mathcal{H}^{\infty}$. If $\psi$ solves Equation~\ref{eqsymbol}, then $H_g$ is the unique Hankel operator from Theorem~\ref{theorem:aakop}.
    \end{corollary}
    In particular, this means that to find the Hankel operator $H_g$ corresponding to the optimal approximation, we can first obtain $\psi$ by applying Corollary \ref{corollary:unimodular}. Then, we can extract the rational component $g$ of $\psi$: this will correspond to a symbol for $\H_g$.

\section{AAK Theory and Approximate Minimization}\label{sec:aakapproxmin}
    
    Theorem~\ref{fliess} establishes a correspondence between a given minimal WFA $A$ with $n$ states and a Hankel matrix $\H$ of rank $n$. The relation between rank and number of states is what motivates our choice to reformulate the approximate minimization problem as low-rank approximation of the Hankel matrix. The approach that we propose to approximate $A$ is to find the minimal WFA corresponding to the Hankel matrix that minimizes $\H$ optimally in the spectral norm. We recall the fundamental result of~\citet{Eckart}.
    \begin{therm}[\cite{Eckart}]\label{thm:eckart}
        Let $\H$ be a Hankel matrix corresponding to a compact Hankel operator of rank $n$, and $\sigma_m$, with $0\leq m <n$ and $\sigma_0 \geq \dots \geq \sigma_{n-1}>0$, its singular numbers. Then, if $\mat{R}$ is a matrix of rank $k$, we have: $\norm{\H - \mat{R}}\geq \sigma_k$. The equality is attained when $\mat{R}$ corresponds to the truncated SVD of $\H$. 
    \end{therm}
    
    In the following example, we compute the low-rank approximation of a finite Hankel matrix using the truncated SVD.  
    \begin{example}
        We consider the Hankel matrix $\mat{M}\in\R^{3\times3}$,
        \begin{equation*}
            \mat{M}=\begin{pmatrix}
                        1 & 2 & 3 \\
                        2 & 3 & 1 \\
                        3 & 1 & 2
        \end{pmatrix}.
        \end{equation*}
        The singular value decomposition of $\mat{M}$ is $\mat{M}=\mat{U}\mat{D}\mat{V}^{\top}$, with
        \begin{equation*}
            \mat{U}= \begin{pmatrix}
                        \frac{1}{\sqrt{3}} & \sqrt{\frac{2}{3}} & 0 \\
                        \frac{1}{\sqrt{3}} & -\frac{1}{\sqrt{6}} & \frac{1}{\sqrt{2}} \\
                        \frac{1}{\sqrt{3}} & -\frac{1}{\sqrt{6}} & -\frac{1}{\sqrt{2}}
        \end{pmatrix}, \quad 
        \mat{D}= \begin{pmatrix}
                        6 & 0 & 0 \\
                        0 & \sqrt{3} & 0 \\
                        0 & 0 & \sqrt{3}
        \end{pmatrix}\quad
        \mat{V}= \begin{pmatrix}
                        \frac{1}{\sqrt{3}} & -\frac{1}{\sqrt{2}} & -\frac{1}{\sqrt{6}} \\
                        \frac{1}{\sqrt{3}} & 0 & \sqrt{\frac{2}{3}} \\
                        \frac{1}{\sqrt{3}} & \frac{1}{\sqrt{2}} & -\frac{1}{\sqrt{6}}
        \end{pmatrix}.
        \end{equation*}
        The rank $2$ matrix $\mat{\overline{M}}$ obtained by truncating the SVD is not Hankel:
        \begin{equation*}
            \mat{\overline{M}}=\begin{pmatrix}
                        1 & 2 & 3 \\
                        \frac{5}{2} & 2 & \frac{3}{2} \\
                        \frac{5}{2} & 2 & \frac{3}{2}
        \end{pmatrix}.
        \end{equation*}
    \end{example}
    It is easy to see that the low-rank approximation of a Hankel matrix obtained by truncating its SVD is not in general a Hankel matrix. This is problematic, since the low-rank approximation needs to be a Hankel matrix in order to correspond to a WFA. On the other hand, we have seen that, by applying AAK theory, we can find the optimal Hankel matrix minimizing the (Hankel) matrix of a Hankel operator in the Hardy spaces. Our objective is to find a way to apply AAK theory to solve the approximate minimization problem of WFAs. To do this, we need an appropriate framework to reformulate this task in terms of Hankel operators and complex functions.

\subsection{Defining a Hankel Operator: the one-letter assumption}\label{aakinterpr}
    
    As a first step, we want to understand whether or not a Hankel operator on the Hardy space can be associated to the Hankel matrix of a weighted automaton. To do so, we compare the Hankel matrix $\H_f$ of a WFA realizing a function $f$ over an alphabet $\Sigma$, to the Hankel matrix $\H_{\phi}$ of a Hankel operator in the Hardy space:
    \begin{equation}
         \H_f = \begin{pmatrix}  f(\varepsilon) & f(a) & f(b) &\dots\\
        f(a) & f(aa)  & f(ab) &\dots \\
                              f(b)  &f(ba) & f(bb) &\dots\\
                              \vdots& \vdots &\vdots&\ddots
            \end{pmatrix} \quad \quad  
    \H_{\phi}=\begin{pmatrix} \widehat{\phi}(-1) & \widehat{\phi}(-2) & \widehat{\phi}(-3) & \dots \\
                               \widehat{\phi}(-2) & \widehat{\phi}(-3) & \widehat{\phi}(-4) &\dots \\
                               \widehat{\phi}(-3) & \widehat{\phi}(-4) & \widehat{\phi}(-5) &\dots \\
                                \vdots& \vdots &\vdots&\ddots
            \end{pmatrix}.
    \end{equation}
    
    We remark that the columns and rows of $\H_f$ are indexed using the letters of the alphabet $\Sigma$:
    \begin{equation*}
        \H_f(p,s) = f(ps) \quad\text{for}\,p,s \in \Sigma,
    \end{equation*}
    while in the case of $\H_{\phi}$, the entries are indexed using natural numbers
    \begin{equation*}
        \H_{\phi}(j,k)= \widehat{\phi}(-j-k-1) \quad\text{for}\, j,k\geq 0.
    \end{equation*}
    If we think of the intuitive definition of the Hankel property presented in the previous section, we have that it holds in both cases the entries of the matrices only depend on the composition of the coordinate. Note that ``composition'' means concatenation of letters in the first case, and sum of numbers in the second one. One fundamental difference is that adding natural numbers is a commutative operation, while concatenating letters is not. For example, while for the matrix corresponding to a Hankel operator in the Hardy space we have:
    \begin{equation*}
        \H_{\phi}(0,1)=\widehat{\phi}(-2)=\H_{\phi}(1,0),    
    \end{equation*}
    in the case of the WFA's matrix, this is not true:
    \begin{equation*}
        \H_f(a,b)=f(ab)\neq f(ba)=\H_f(b,a).    
    \end{equation*}
    This fact reflects in the much stronger structural property satisfied by Hankel of matrices in the Hardy spaces, where the Hankel property implies that the anti-diagonals have constant entries. This property is not reflected by the matrix of an arbitrary WFA, so it is not always possible to associate a Hankel operator to an automaton over an alphabet of arbitrary size, and AAK theory cannot be generally applied. The only case in which concatenation of strings is commutative, is when we are restricting our focus on alphabets of one letter. In particular, when $|\Sigma|=1$, the set of strings $\Sigma^*$ can be identified with $\N$. Therefore, the function $f:\Sigma^* \rightarrow \R$ recognized by a minimal WFA can be rewritten as $f:\N \rightarrow \R$, and the Hankel matrix $\H_f$ associated with it can be interpreted as the matrix of a Hankel operator between sequences $H_f:\ell^2\rightarrow\ell^2$. In this case, the Hankel matrix is defined by $\H(i,j)=f(i+j)$, for $i,j\geq 0$:
    \begin{equation*}
        \H_f= \begin{pmatrix} f(0) & f(1) & f(2) & \dots \\
                              f(1) & f(2) & f(3) &\dots \\
                              f(2) & f(3) & f(4) &\dots \\
                                \vdots& \vdots &\vdots&\ddots
            \end{pmatrix}.
    \end{equation*}
    Using the Fourier isomorphism, we can interpret $\H_f$ as the matrix $\H_{\phi}$ of a Hankel operator over Hardy spaces, associated to a complex function $\phi\in  \mathcal{L}^2(\mathbb{T})$. In particular, we can embed the sequence space $\ell^2$ into $\ell^2(\Z)$ by ``duplicating'' each vector, \emph{i.e.} by associating $\boldsymbol{\mu}=(\mu_0, \mu_1, \dots)\in\ell^2$ to $\boldsymbol{\mu}^{(2)}=(\dots, \mu_1,\mu_0, \mu_1, \dots)\in\ell^2(\Z)$. Then, we can use the Fourier isomorphism to map the vector $\boldsymbol{\mu}^{(2)}\in\ell^2(\Z)$ to the complex function space $\mathcal{L}^2(\mathbb{T})$. In this way, each vector $\boldsymbol{\mu}\in \ell^2$ corresponds to two functions in the Hardy spaces: 
    \begin{align}\label{eq:hardynotation}
        &\mu^-(z)=\sum_{j=0}^{\infty}\boldsymbol{\mu}_j z^{-j-1} \in \mathcal{H}^2_-, \\
        &\mu^+(z)=\sum_{j=0}^{\infty} \boldsymbol{\mu}_j z^{j} \in \mathcal{H}^2.\notag
    \end{align}
    %The entries of the matrix are defined by means of the Fourier coefficients of $\phi$ as $\H(j,k)= \widehat{\phi}(-j-k-1)$ for $j,k\geq 0$.
    %Note that the function $\mathbb{P}_-\phi=\widehat{\phi}(-j-k-1)$ is a complex rational function \cite{kronecker}. 
    % The matrix $\H$ can be represented as:
    % \begin{equation*}
    %     \H=\H_{\phi}=\begin{pmatrix} \widehat{\phi}(-1) & \widehat{\phi}(-2) & \widehat{\phi}(-3) & \dots \\
    %                           \widehat{\phi}(-2) & \widehat{\phi}(-3) & \widehat{\phi}(-4) &\dots \\
    %                           \widehat{\phi}(-3) & \widehat{\phi}(-4) & \widehat{\phi}(-5) &\dots \\
    %                             \vdots& \vdots &\vdots&\ddots
    %         \end{pmatrix}.
    % \end{equation*}
    Moreover, we can derive the relationship between $f$ and $\phi$: 
    \begin{equation*}
        \begin{pmatrix} f_A(0) & f_A(1) & f_A(2) & \dots \\
                              f_A(1) & f_A(2) & f_A(3) &\dots \\
                              f_A(2) & f_A(3) & f_A(4) &\dots \\
                                \vdots& \vdots &\vdots&\ddots
            \end{pmatrix}
            =\begin{pmatrix} \widehat{\phi}(-1) & \widehat{\phi}(-2) & \widehat{\phi}(-3) & \dots \\
                              \widehat{\phi}(-2) & \widehat{\phi}(-3) & \widehat{\phi}(-4) &\dots \\
                              \widehat{\phi}(-3) & \widehat{\phi}(-4) & \widehat{\phi}(-5) &\dots \\
                                \vdots& \vdots &\vdots&\ddots
            \end{pmatrix}.
    \end{equation*}
    from which we obtain:
    \begin{equation}\label{eq:symbol}
        f(n)=\widehat{\phi}(-n-1).    
    \end{equation}
    
    Since we know how to express the function $f$ with  respect to the parameters of the WFA, we can explicitly compute the rational component of the symbol:
    \begin{equation}\label{eq:symbolsum}
        \mathbb{P}_-\phi= \sum_{k\geq 0}f(k) z^{-k-1} = \sum_{k\geq 0} \balpha^{\top}\A^k \bbeta z^{-k-1} = \balpha^{\top}(z\mat{1}-\A)^{-1} \bbeta,
    \end{equation}
    where the last equality holds only if $\rho(A)<1$.
    
    The correspondence between symbol and function computed by a model allows us to reformulate the approximation problem in terms of Hankel operators and functions in the complex space, and to apply AAK theory. 

    We consider the following example, from~\cite{AAK-WFA}. 
    
    \begin{example}\label{example1}
    
    Let $|\Sigma|=1$, $\Sigma=\{x\}$, we consider the WFA $A= \langle \balpha, \A, \bbeta\rangle$ represented in Figure~\ref{fig:wfaex}, with: 
        \begin{equation*}
            \A= \begin{pmatrix}
                    0 & \frac{1}{2}    \\
                    \frac{1}{2}  & 0  
                \end{pmatrix} , \quad
            \balpha= \begin{pmatrix}
                \frac{\sqrt{3}}{2}     \\
                0  
                \end{pmatrix}, \quad
            \bbeta= \begin{pmatrix}
                \frac{\sqrt{3}}{2}     \\
                0  
                \end{pmatrix},
        \end{equation*}
    
    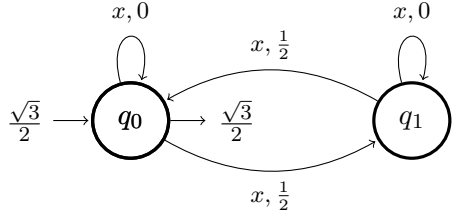
\begin{figure}[t]
	    \centering
	    %\begin{subfigure}[l]{0.4\textwidth}
	    	%\centering
		    \begin{tikzpicture}[shorten >=1pt,node distance=3.7cm,auto]
		        \tikzstyle{every state}=[draw=black,very thick,fill=none,minimum size=10mm]
		        \tikzstyle{accepting}=[accepting by arrow]
		        \node[state, initial,initial text=$\frac{\sqrt{3}}{2}$] (q_0) {$q_0$};
		        \node[state, accepting, accepting text=$\frac{\sqrt{3}}{2}$] (q_0) {$q_0$};
		        \node[state] (q_1) [right of=q_0] {$q_1$};
		        \path[->]
		        (q_0) edge [loop above] node
		        {\footnotesize $\begin{matrix}{{x , 0}}\end{matrix}$} ()
		        (q_0) edge [bend right] node [swap]
		        {\footnotesize $\begin{matrix}{{x , \frac{1}{2}}}\end{matrix}$} (q_1)
		        (q_1) edge [loop above] node
		        {\footnotesize $\begin{matrix}{{x , 0}}\end{matrix}$} ()
		        (q_1) edge [bend right] node [swap]
		        {\footnotesize $\begin{matrix}{{x , \frac{1}{2}}}\end{matrix}$} (q_0);
		    \end{tikzpicture}
	\caption[Example of generative probabilistic automaton]{ Graphical representation of the generative probabilistic automaton described in Example~\ref{example1}.}
	    \label{fig:wfaex} 
    \end{figure}
    
    Note that $A$ is a generative probabilistic automaton. Indeed, we have that 
    \begin{itemize}
        \item $f_A(x)\geq 0$
        \item $\sum_{x\in\Sigma^*}f_A(x)=1$,
    \end{itemize}
    since the rational function realized by the WFA is defined as:
        \begin{equation*}
            f_A(x\cdots x)=f_A(k)= \balpha^{\top}\A^k\bbeta= \begin{cases}
                                                0   &\text{if $k$ is odd} \\
                                                \frac{3}{4}2^{-k}   &\text{if $k$ is even}
                                                \end{cases}
        \end{equation*}
    where $k$ corresponds to the string where $x$ is repeated $k$-times.
    We remark that $A$ is minimal and already in its SVA form, with Gramians 
    \begin{equation}
        \mP=\mQ= \begin{pmatrix}
                    \frac{4}{5}  & 0  \\
                    0            & \frac{1}{5} 
                \end{pmatrix}.
    \end{equation}
    The corresponding Hankel matrix, with entries defined as $\H(i,j)=f(i+j)$, has rank $2$:
        \begin{equation}
            \H=\begin{pmatrix} f_A(0) & f_A(1) & f_A(2) & \dots \\
                               f_A(1) & f_A(2) & f_A(3) &\dots \\
                               f_A(2) & f_A(3) & f_A(4) &\dots \\
                                \vdots& \vdots &\vdots&\ddots
            \end{pmatrix}=
            \begin{pmatrix} \frac{3}{4} & 0 & \frac{3}{16} & \dots \\
                               0 & \frac{3}{16} & 0 &\dots \\
                               \frac{3}{16} & 0 & \frac{3}{64} &\dots \\
                                \vdots& \vdots &\vdots&\ddots
            \end{pmatrix}.
        \end{equation} 
    Now, we can apply the second interpretation of the Hankel matrix and look at it with respect to the symbol, using the definition $\H(j,k)= \widehat{\phi}(-j-k-1)$. We have:
        \begin{equation*}
             \H=
            \begin{pmatrix} \frac{3}{4} & 0 & \frac{3}{16} & \dots \\
                               0 & \frac{3}{16} & 0 &\dots \\
                               \frac{3}{16} & 0 & \frac{3}{64} &\dots \\
                                \vdots& \vdots &\vdots&\ddots
            \end{pmatrix}
            =\begin{pmatrix} \widehat{\phi}(-1) & \widehat{\phi}(-2) & \widehat{\phi}(-3) & \dots \\
                               \widehat{\phi}(-2) & \widehat{\phi}(-3) & \widehat{\phi}(-4) &\dots \\
                               \widehat{\phi}(-3) & \widehat{\phi}(-4) & \widehat{\phi}(-5) &\dots \\
                                \vdots& \vdots &\vdots&\ddots
            \end{pmatrix}.
        \end{equation*} 
    We can recover the rational component of a symbol, \emph{i.e.} the projection of $\phi$ on the negative Hardy space.
    \begin{equation*}
        \mathbb{P}_-\phi=\sum_{n \geq 0}\widehat{\phi}(-n-1)z^{-n-1}=\sum_{n \geq 0}\frac{3}{4}4^{-n}z^{-2n-1}=\frac{3z}{4z^2-1}.
    \end{equation*}
    Note that this is a complex rational function having degree $2$, and it has two poles inside the unit disc at $z=\pm\frac{1}{2}$ (as predicted by Theorem~\ref{theorem:Kronecker}). It is important to remark that from the Hankel matrix we can only recover the negative Fourier coefficients of $\phi$, meaning only the component of the symbol that belongs to the negative Hardy space.
    \end{example}

\section{Solving the Approximate Minimization Problem}\label{solution}

In this section we present the theoretical contribution of this paper, a closed form solution for the approximate minimization problem.

\subsection{Assumptions}\label{assumptions}
    
    We briefly list and analyze the assumptions made to solve the approximate minimization problem. A class of automata that automatically satisfies the following properties is that of generative probabilistic automata.

    \paragraph*{One-letter Alphabet}
    
    We tackle the approximate minimization problem in the case of automata with real weights, defined over a one-letter alphabet. As discussed before, this assumption is needed in order to apply AAK theory, and will hold for the rest of the paper (see Section \ref{aakinterpr} for more details).
    
    \paragraph*{SVA Form}    
    
    We assume that the minimal WFA $A= \langle \balpha, \A, \bbeta\rangle$ is in SVA form. This assumption is not necessary, as the SVA can be efficiently computed from a WFA satisfying the set of assumptions stated above \citep{Balle19}. Starting from a WFA in SVA form allows us to obtain results that are representation independent. Since the alphabet has size one, the Hankel matrix $\H$ is symmetric. Therefore, if we denote with $\lambda_i$ the $i$-th non-zero eigenvalue of $\H$, and we consider the coordinates of $\balpha$ and $\bbeta$, we have that $\balpha_i=\operatorname{sgn}(\lambda_i)\bbeta_i$, where $\operatorname{sgn}(\lambda_i)=\lambda_i/|\lambda_i|$.
    
    \paragraph*{Compactness of the Operator}
    
    To apply Theorem~\ref{theorem:aakop} we need the Hankel operator $H$ to be compact. To ensure that this condition is satisfied, we study the respective Hankel matrix.
    In our setting, the Hankel matrix has finite rank (equal to the number of states of the minimal WFA that we are considering). Moreover, the singular values can be computed exactly using the Gramian matrices introduced in Definition \ref{def:gramians}. A finite rank operator is compact only if it is bounded. Therefore, we just need to check that the Hankel operator is bounded. To this extent, \citet{Balle19} show that it is enough that the WFA being considered computes a function $f\in\ell^2$. We make the slightly stronger assumption that  the transition matrix $A$ is irreduntant, \emph{i.e.} that $\rho(\A)<1$, where $\rho$ is the spectral radius. This condition directly implies boundness and the existence of the SVA and the Gramian matrices $\mP$ and $\mQ$, where $\mP=\mQ$ and are diagonal matrices \citep{Balle19}. Moreover, it allows us to compute a closed form for the symbol of a WFA, as seen in Equation \ref{eq:symbolsum}.

\subsection{Problem Formulation}

    Let $A= \langle \balpha, \A, \bbeta \rangle$ be a minimal irredundant WFA with $n$ states and real weights, defined over a one-letter alphabet and represented in its SVA form. Let $\H$ be the Hankel matrix of $A$, we denote with $\sigma_i$, for $0\leq i < n$, the singular numbers. Given a target number of states $k<n$, we say that a WFA $\widehat{A}_k$ with $k$ states solves the \emph{optimal spectral-norm approximate minimization} problem if the Hankel matrix $\mat{G}$ of $\widehat{A}_k$ satisfies:
    \begin{equation}
        \norm{\H - \mat{G}}= \sigma_k(\H).
    \end{equation}
    
    Note that the content of the ``optimal spectral-norm approximate minimization'' is equivalent to the problem solved by Theorem \ref{theorem:aakop}, with the exception that here we represent the inputs and outputs of the problem effectively by means of WFAs. 
    
    Based on the AAK theory sketched in Section~\ref{aak-section}, we draw the following steps:
    
    \begin{itemize}
        \item[(1)] \emph{Compute a symbol for the WFA}. Given an irredundant WFA on a one-letter alphabet, we consider its Hankel matrix $\H$ and the function $f$ that it is computing. We use Equation \ref{eq:symbolsum} to associate a complex rational function to the WFA.
        \item[(2)]\label{step2} \emph{Compute the optimal symbol $\psi(z)$ using Corollary~\ref{corollary:unimodular}}. The main challenge here is to find a suitable representation for the functions $\psi(z)$ and $e(z)=\phi(z)-\psi(z)$. We define them in terms of two auxiliary WFAs. The key point is to select constraints on their parameters to leverage the properties of weighted automata, while still keeping the formulation general.
        \item[(3)] \emph{Extracting the rational component by solving for $g(z)$ in Corollary~\ref{corollary:stable}}. This step is arguably the most conceptually challenging, as it requires to identify the position of the function's poles. In fact, we know from Theorem~\ref{theorem:Kronecker} that $g(z)$ has $k$ poles, all inside the unit disc.
        \item[(4)] \emph{Find a WFA representation for $g(z)$}. Since in Step~\ref{step2} we parametrized the functions using WFAs, the expression of $g(z)$ directly reveals the WFA $\widehat{A}_k$.
    \end{itemize}

\subsection{Finding the Optimal Approximation}
    
We analyze each of the steps detailed above.

\subsubsection{Finding a Symbol for the WFA}\label{sec:symbolWFA}

    Let $A= \langle \balpha, \A, \bbeta \rangle$ be a minimal irredundant WFA with $n$ states, already represented in its in SVA form. $A$ realizes a function $f: \Sigma^* \rightarrow \R $, defined over a one-letter alphabet $\Sigma=\{a\}$. Let $\H$ be its Hankel matrix, with corresponding bounded Hankel operator $H$, and singular numbers $\sigma_i$, for $0\leq i < n$. 
    
    As seen in Subsection \ref{aakinterpr}, we can associate a complex function to the WFA. In particular, since we are assuming that $A$ is irredundant, from Equation \ref{eq:symbolsum} we obtain an expression for the rational component of the symbol:
    \begin{equation}
        \mathbb{P}_-\phi= \balpha^{\top}(z\mat{1}-\A)^{-1} \bbeta.
    \end{equation}
    % From the theory of Hankel operators, we know that we can choose as symbol for the Hankel operator any complex function whose component in the negative Hardy space is equal to $\mathbb{P}_-\phi$. Moreover, since we know that the operator is bounded, we can choose symbol within the class of bounded functions without loosing any generality (Theorem \ref{thm:nehari}). Since $\mathbb{P}_-\phi$ is a bounded function, we can choose it to be the operator's symbol, without adding any component from the positive Hardy space. Thus, with an abuse of notation, we simply consider $\phi=\balpha^{\top}(z\mat{1}-\A)^{-1} \bbeta$ as a symbol for $H$.

\subsubsection{Finding the Optimal Symbol}

    To find the solution of Theorem \ref{theorem:aakop}, we need to first derive the function $\psi$ from Theorem \ref{theorem:aaksymb}. Therefore, the second step to solve the approximate minimization problem is to find a proper expression for the complex functions $\psi$ and $e=\phi-\psi$ described in Theorem \ref{theorem:aaksymb}. Since our objective is to find the WFA corresponding to the optimal approximation, we focus on representing these functions using the parameters of two auxiliary WFAs.
    %This allows us to leverage at the same time properties of WFAs and of complex functions.
    We consider a WFA $\widehat{A}=\langle\widehat{\balpha},\widehat{\A},\widehat{\bbeta}\rangle$ with more than $k$ states, such that the automaton $E= \langle \balpha_e, \A_e, \bbeta_e \rangle$ computing the difference between $A$ and $\widehat{A}$ is minimal, with:
     \begin{equation}\label{eq:partition.Ae}
            \A_e= \begin{pmatrix}
                    \A & \mat{0}    \\
                    \mat{0}  & \widehat{\A}  
                \end{pmatrix} , \quad
            \balpha_e= \begin{pmatrix}
                \balpha          \\
                -\widehat{\balpha}  
                \end{pmatrix}, \quad
            \bbeta_e= \begin{pmatrix}
                \bbeta          \\
                \widehat{\bbeta}  
                \end{pmatrix}.
        \end{equation}
    Now, given $C\in\mathcal{H}^{\infty}$, we consider the complex functions:
    \begin{align*}
        &\psi= \widehat{\balpha}^{\top}(z\mat{1}-\widehat{\A})^{-1} \widehat{\bbeta} + C\\
        &e= \phi-\psi= \balpha_e^{\top}(z\mat{1}-\A_e)^{-1} \bbeta_e - C.
    \end{align*}
    % Since our objective is to find the WFA corresponding to the optimal approximation, we focus on representing these functions using the parameters of two auxiliary WFAs. This allows us to leverage at the same time properties of WFAs and of complex functions. We define the WFAs as follows. Let $k$ be the target number of states of the best approximation, we consider a WFA $\widehat{A}=\langle\widehat{\balpha},\widehat{\A},\widehat{\bbeta}\rangle$ with $j\geq k$ states,
    % satisfying the following two properties:
    % \begin{itemize}
    %     \item[(1)] $1$ is not an eigenvalue of $\widehat{\A}$
    %     \item[(2)]\label{Eminimal} the automaton  $E= \langle \balpha_e, \A_e, \bbeta_e \rangle$ is minimal, where
    %     \begin{equation}\label{eq:partition.Ae}
    %         \A_e= \begin{pmatrix}
    %                 \A & \mat{0}    \\
    %                 \mat{0}  & \widehat{\A}  
    %             \end{pmatrix} , \quad
    %         \balpha_e= \begin{pmatrix}
    %             \balpha          \\
    %             -\widehat{\balpha}  
    %             \end{pmatrix}, \quad
    %         \bbeta_e= \begin{pmatrix}
    %             \bbeta          \\
    %             \widehat{\bbeta}  
    %             \end{pmatrix}.
    %     \end{equation}
    % \end{itemize}
    % Using the parameters of the automaton $\widehat{A}$, and a constant $C$, we define a function $\psi= \widehat{\balpha}^{\top}(z\mat{1}-\widehat{\A})^{-1} \widehat{\bbeta} + C$. Analogously, we use the parameters of $E$ to define the function $e= \phi-\psi= \balpha_e^{\top}(z\mat{1}-\A_e)^{-1} \bbeta_e - C$. 
    The idea is that we want to find the parameters of $\widehat{A}$ that make $\psi$ the solution of Theorem \ref{theorem:aaksymb}. % The first important condition which needs to be satisfied is the boundness around the unit circle.
    By definition, $\psi$ is the sum of two components, one that is bounded around the unit circle and one that has $k$ poles inside the unit disc (where $k$ is the size of the sought approximation). Therefore, there cannot be poles on the unit circle. By looking at the way we defined the function $\psi$, we can see that its poles correspond to the eigenvalues of $\widehat{\A}$, counted with their multiplicities. Thus, in order for $\psi$ to be the solution of Theorem \ref{theorem:aaksymb}, $1$ cannot be an eigenvalue of $\widehat{\A}$, and the WFA $\widehat{A}$ needs to have at least $k$ states. 
    % Thus, we assume that this property holds for the WFA $\widehat{A}$. Note that, since the spectral radius of $\A$ is smaller then one, also $e= \phi-\psi= \balpha_e^{\top}(z\mat{1}-\A_e)^{-1} \bbeta_e - C$ will be bounded around the unit circle. 

    As remarked in the previous section, the parameters of the automaton $A$ only encode the negative Fourier coefficients of the symbol. We add $C$ to the definition of $\psi$ to account for the $\mathcal{H}^{\infty}$ component when considering the difference $\phi-\psi$. In fact, while this component of the symbol does not affect the spectral norm, it plays a role in the computation of the $\mathcal{L}^{\infty}$-norm (in Equation~\ref{eqsymbol}), so it cannot be entirely dismissed. Nonetheless, we won't need to find the value of $C$, as ultimately we are only interested in the WFA's parameters.

    % Since the first condition is satisfied, we can use Corollary~\ref{corollary:unimodular} to find the triple $\widehat{\balpha},\widehat{\A},\widehat{\bbeta}$ such that $\psi$ satisfies Equation~\ref{eq:unimodular}. 
    Now that we have an expression for $\psi$ and $e$, we can look back at Theorem \ref{theorem:aaksymb}. From this theorem, we know that by definition, $\sigma_k^{\-1}e$ is a unimodular function. 
    In Appendix \ref{apd:first} we show how to compute the functions corresponding to  a $\sigma_k$-Schmidt pair, and that their quotient is indeed unimodular. 
    This property of $e$ can be used to derive a set of constraints on the parameters of the WFA $E=\langle \balpha_e, \A_e, \bbeta_e\rangle$. In particular, it is possible to leverage the maximum modulus principle, according to which the maximum modulus of an holomorphic function is attained on the boundary of the domain. A similar result has been obtained in the control theory literature \citep{discreteH}, and can be easily applied to our setting. In fact, a parallel can be drawn between dynamical systems and automata, by noting that the impulse-response of a discrete time-invariant Single-Input-Single-Output SISO system can be parametrized as a WFA over a one-letter alphabet. This allows us to apply a theorem from \citet[Theorem 6.3]{discreteH} to find two matrices, $\mP_e$ and $\mQ_e$, satisfying properties similar to those of the Gramians. It is important to notice that, \emph{a priori}, the controllability and observability Gramians of $E$ might not be well defined. 
    % The proof of the following theorem relies on the minimality of the WFA $E$~\cite{schutter}. 
    We refer the reader to Appendix~\ref{apd:first} for a sketch of the proof.
    
    \begin{therm}[\cite{discreteH}]\label{Theorem:allpass}
        Consider the function $e= \balpha_e^{\top}(z\mat{1}-\A_e)^{-1} \bbeta_e - C$ and the corresponding minimal WFA $E=\langle \balpha_e, \A_e, \bbeta_e \rangle$ associated with it. If $\sigma_k^{-1}e$ is unimodular, then there exists a unique pair of symmetric invertible matrices $\mP_e$ and $\mQ_e$ satisfying:
        \begin{itemize}
            \item[(1)] $\mP_e-\A_e \mP_e \A_e^\top = \bbeta_e\bbeta_e^\top$ \label{a}
            \item[(2)] $\mQ_e-\A_e^\top \mQ_e \A_e = \balpha_e\balpha_e^\top$ \label{b}
            \item[(3)] $\mP_e\mQ_e=\sigma^2_k\mat{1}$ \label{c}
        \end{itemize}
    \end{therm}

    We can now derive the parameters of the WFA $\widehat{A}=\langle\widehat{\balpha},\widehat{\A},\widehat{\bbeta}\rangle$ that make $\psi$ the solution of Theorem~\ref{theorem:aaksymb}.
    
    \begin{therm}\label{Theorem:maintrm}
    Let $A=\langle \balpha, \A, \bbeta \rangle $ be a minimal WFA with $n$ states in its SVA form, and let $\phi= \balpha^{\top}(z\mat{1}-\A)^{-1}\bbeta$ be a symbol for its Hankel operator $H$. Let $\sigma_k$ be a singular number of multiplicity $r$ for $H$, with:
    \begin{equation}
        \sigma_0 \geq \dots > \sigma_k=\dots=\sigma_{k+r-1}>\sigma_{k+r}\geq \dots \geq \sigma_{n-1}>0.
    \end{equation}
    We can partition the Gramian matrices $\mP$, $\mQ$ as follows:
    \begin{equation}\label{eq:setup}
        \mP=\mQ=\begin{pmatrix}
            \bm{\Sigma}  & \mat{0}     \\
            \mat{0}      & \sigma_k\mat{1}_r
            \end{pmatrix},
    \end{equation}
    where $\boldsymbol{\Sigma}\in \R^{(n-r)\times(n-r)}$ is the diagonal matrix containing the remaining singular numbers, and partition $\A$, $\balpha$ and $\bbeta$ to conform with the Gramians:
    \begin{equation}\A= \begin{pmatrix}
            \A_{11}          & \A_{12} \\
            \A_{21}   & \A_{22} 
            \end{pmatrix},\quad
        \balpha= \begin{pmatrix}
             \balpha_1 \\
             \balpha_2 
            \end{pmatrix},\quad
        \bbeta= \begin{pmatrix}
             \bbeta_1 \\
             \bbeta_2 
            \end{pmatrix}.
    \end{equation}
    Let $\mat{R}=\sigma_k^2\mat{1}_{n-r}-\bm{\Sigma}^2$, we denote by $(\cdot)^{+}$ the Moore-Penrose pseudo-inverse. If the function $\psi= \widehat{\balpha}^{\top}(z\mat{1}-\widehat{\A})^{-1} \widehat{\bbeta}+C$ is the best approximation of $\phi$, then:
    \begin{itemize}
        \item  If $\balpha_2 \neq \mat{0}$:
        \begin{equation}\label{eq:icalpsol}
            \begin{cases}
                \widehat{\bbeta} = - \widehat{\A}\A_{21}^{\top}(\bbeta_2^{\top})^{+} \\
                \widehat{\balpha} = \widehat{\A}^{\top}\mat{R}\A_{12}(\balpha_2^{\top})^{+}\\
                \widehat{\A}(\A_{11}^{\top}- \A_{21}^{\top}(\bbeta_2^{\top})^{+}\bbeta_1^{\top})=\mat{1}
        \end{cases}
        \end{equation}
        \item If $\balpha_2=\mat{0}$:
        \begin{equation}\label{eq:balphazero}
            \begin{cases}
                \widehat{\bbeta} =  (\mat{1} -\widehat{\A}\A_{11}^{\top})(\bbeta_1^{\top})^{+}\\
                \widehat{\balpha} =-(\mat{R} -\widehat{\A}^{\top}\mat{R}\A_{11})(\balpha_1^{\top})^{+}\\
                \widehat{\A}\A_{21}^{\top}=\mat{0}
            \end{cases}
        \end{equation}
    \end{itemize}
    \end{therm}
    
    \begin{proof}
    Since $\sigma^{-1}e=\phi-\psi$ is unimodular, from Theorem~\ref{Theorem:allpass} there exist two symmetric nonsingular matrices $\mP_e$, $\mQ_e$ satisfying the fixed point equations:
        \begin{align}\label{fixed_point_e}
            \mP_e-\A_e\mP_e\A_e^{\top}&=\bbeta_e\bbeta_e^{\top}\\
            \mQ_e-\A_e^{\top}\mQ_e\A_e&=\balpha_e\balpha_e^{\top}\label{fixed_point_e2}
        \end{align}
    and such that $\mP_e\mQ_e=\sigma^2_k\mat{1}$.
    We can partition $\mP_e$ and $\mQ_e$ according to the definition of $\A_e$~(see Equation~\ref{eq:partition.Ae}):
        \begin{equation*}
            \mP_e= \begin{pmatrix}
                \mP_{11}          & \mP_{12} \\
                \mP_{12}^{\top}   & \mP_{22} 
                \end{pmatrix},\quad
            \mQ_e= \begin{pmatrix}
                 \mQ_{11}          & \mQ_{12} \\
                \mQ_{12}^{\top}   & \mQ_{22} 
          \end{pmatrix}.
        \end{equation*}
    From Equation~\ref{fixed_point_e} and~\ref{fixed_point_e2}, we note that $\mP_{11}$ and $\mQ_{11}$ correspond to the controllability and observability Gramians of $A$: 
        \begin{equation*}
            \mP_{11}=\mQ_{11}=\mP=\begin{pmatrix}
                          \bm{\Sigma} & \mat{0}     \\
                          \mat{0}      & \sigma_k\mat{1}
                          \end{pmatrix}.
        \end{equation*}
    Moreover, since $\mP_e\mQ_e=\sigma_k^2\mat{1}$, we get $\mP_{12}\mQ_{12}^{\top}=\sigma_k^2\mat{1}-\mP^2$. It follows that $\mP_{12}\mQ_{12}^{\top}$ has rank $n-r$. Without loss of generality we can set $\dim{\widehat{\A}}=j=n-r$, and choose an appropriate basis for the state space such that $\mP_{12}=\begin{pmatrix} \mat{1} &  \mat{0} \end{pmatrix} ^{\top}$ and $\mQ_{12}=\begin{pmatrix} \mat{R} &  \mat{0} \end{pmatrix} ^{\top}$, with $\mat{R}=\sigma_k^2\mat{1}-\bm{\Sigma}^2$. Once $\mP_{12}$ and $\mQ_{12}$ are fixed, the values of $\mP_{22}$ and $\mQ_{22}$ are automatically determined. We obtain:
    \begin{equation}
        \mP_e= \begin{pmatrix}
            \bm{\Sigma}     & \mat{0}                    & \mat{1}\\
            \mat{0}          & \sigma_k\mat{1}     & \mat{0}\\
             \mat{1} & \mat{0}                    & -\bm{\Sigma} \mat{R}^{-1}  
             \end{pmatrix},\quad
        \mQ_e= \begin{pmatrix}
            \bm{\Sigma}     & \mat{0}                    & \mat{R}\\
            \mat{0}          & \sigma_k\mat{1}     & \mat{0}\\
            \mat{R}         & \mat{0}                    & -\bm{\Sigma} \mat{R}
            \end{pmatrix}.
    \end{equation}
    
    Now that we have an expression for the matrices $\mP_e$ and $\mQ_e$ of Theorem~\ref{Theorem:allpass}, we can  rewrite the fixed point equations to derive the parameters $\widehat{\balpha}$, $\widehat{\A}$ and $\widehat{\bbeta}$. We obtain the following systems:
        \begin{equation}
            \begin{cases}
                \mP -\A \mP \A^{\top}= \bbeta\bbeta^{\top}\\
                \mat{N} -\A \mat{N} \widehat{\A}^{\top}= \bbeta\widehat{\bbeta}^{\top}\\
                -\bm{\Sigma} \mat{R}^{-1} +\widehat{\A}\bm{\Sigma} \mat{R}^{-1}\widehat{\A}^{\top}={\widehat{\bbeta}}{\widehat{\bbeta}}^{\top}
             \end{cases} \quad  \begin{cases}
                \mP -\A^{\top} \mP \A= \balpha\balpha^{\top}\\
                \mat{M} -\A^{\top} \mat{M} \widehat{\A}= -\balpha\widehat{\balpha}^{\top}\\
                -\bm{\Sigma} \mat{R} +\widehat{\A}^{\top}\bm{\Sigma} \mat{R}\widehat{\A}={\widehat{\balpha}}{\widehat{\balpha}}^{\top}
            \end{cases}
        \end{equation}
    where $\mat{N}= \begin{pmatrix}
                \mat{1}     \\
                \mat{0}
            \end{pmatrix}$ and $\mat{M}= \begin{pmatrix}
                                         \mat{R}     \\
                                         \mat{0}
                                          \end{pmatrix}$. 
                                          
    We can rewrite the second equation of each system as follows:
        \begin{equation}
            \begin{cases}
                \mat{1} -\A_{11} \widehat{\A}^{\top}= \bbeta_1\widehat{\bbeta}^{\top}\\
                -\A_{21}\widehat{\A}^{\top}=\bbeta_2\widehat{\bbeta}^{\top}
             \end{cases} \quad  \begin{cases}
                \mat{R} -\A_{11}^{\top} \mat{R} \widehat{\A}=-\balpha_1\widehat{\balpha}^{\top}\\
                \widehat{\A}^{\top}\mat{R}\A_{12}=\widehat{\balpha}\balpha_2^{\top}
            \end{cases}
        \end{equation}        
    If $\balpha_2 \neq \mat{0}$, then also $\bbeta_2 \neq \mat{0}$ (recall that $\balpha_i=\operatorname{sgn}(\lambda_i)\bbeta_i$), and we have:                        
    \begin{equation}
        \begin{cases}
                \widehat{\bbeta} = - \widehat{\A}\A_{21}^{\top}(\bbeta_2^{\top})^{+} \\
                \widehat{\balpha} = \widehat{\A}^{\top}\mat{R}\A_{12}(\balpha_2^{\top})^{+}\\
                \widehat{\A}(\A_{11}^{\top}- \A_{21}^{\top}(\bbeta_2^{\top})^{+}\bbeta_1^{\top})=\mat{1}
        \end{cases}
    \end{equation}
    with $(\balpha_2^{\top})^{+}=\frac{\balpha_2}{\balpha_2^{\top}\balpha_2}$ and $(\bbeta_2^{\top})^{+}=\frac{\bbeta_2}{\bbeta_2^{\top}\bbeta_2}$.
    
    If $\balpha_2=\mat{0}$, we have $\widehat{\A}\A_{21}^{\top}=\mat{0}$. We remark that $\widehat{\A}$ has size $(n-r)\times(n-r)$, while $\A_{21}^{\top}$ is $(n-r)\times r$, so the system of equations corresponding to $\widehat{\A}\A_{21}^{\top}=\mat{0}$ is underdetermined if $r<\frac{n}{2}$, in which case we can find an alternative set of solutions:
    \begin{equation}\label{eq:case2}
            \begin{cases}
                \widehat{\bbeta} =  (\mat{1} -\widehat{\A}\A_{11}^{\top})(\bbeta_1^{\top})^{+}\\
                \widehat{\balpha} =-(\mat{R} -\widehat{\A}^{\top}\mat{R}\A_{11})(\balpha_1^{\top})^{+}\\
                \widehat{\A}\A_{21}^{\top}=\mat{0}
            \end{cases}
    \end{equation}
    with $\widehat{\A}\neq \mat{0}$. On the other hand, if $r\geq\frac{n}{2}$, \emph{i.e.} if the multiplicity of the singular number $\sigma_k$ is more than half the size of the original WFA, the system might not have any solution unless $\widehat{\A}=\mat{0}$ (or unless $\A_{21}$ was zero to begin with). In this setting the method proposed returns $\widehat{\A}=\mat{0}$. 
    \end{proof} 
    
    We remark that in the (rare) case in which the algorithm returns $\widehat{\A}=\mat{0}$, an alternative and preferable approach is to search for an approximation of size $k-1$ or $k+1$. This way, the multiplicity $r$ of the singular number $\sigma_k$ is such that $r<\frac{n}{2}$, and the system in Equation~\ref{eq:case2} is underdetermined.
    
    Theorem~\ref{Theorem:maintrm} provides us with a way to compute the coefficients of the function $\psi$ solving Theorem~\ref{theorem:aaksymb}. It is important to notice that the WFA $A_k$ is not necessarily the best approximation we are looking for. Intuitively, the problem is that it might be too big, as irredundancy is not guaranteed by the system of equations (while we know from AAK theory that the best approximation corresponds to a bounded operator). Therefore, in these cases we need to ``extract'' from $A_k$ a smaller WFA of size $k$. We do this by extracting the component of the function $\psi$ that belongs to the negative Hardy space.

\subsubsection{Extracting the Rational Component}\label{sec:rational_comp}

     The objective of this section is to ``isolate'' the function $g\in \mathcal{R}_k$, \emph{i.e.} the \emph{rational component} of $\psi$.
    To do this, we study the position of the poles of $\psi$. In fact, we know from Theorem~\ref{theorem:Kronecker} that the poles of a strictly proper rational function lie inside the unit disc. As noted before, the key to solving our problem is the way we parametrized the functions. We defined $\psi$ so that its poles correspond to the eigenvalues of $\widehat{A}$. Therefore, we study the eigenvalues of $\widehat{\A}$ using the following auxiliary result from~\citet{ostrowski}. A proof of this theorem can be found in~\cite{inertiaproof}.
    \begin{therm}[\cite{ostrowski}]\label{theorem:inertia}
        Let $|\Sigma|=1$, and let $\mP$ be a solution to the fixed point equation $X-\A X\A^{\top}=\bbeta\bbeta^{\top}$ for the WFA $A=\langle \balpha, \A, \bbeta \rangle$. If $A$ is reachable, then:
        \begin{itemize}
            \item The number of eigenvalues $\lambda$ of $\A$ such that $|\lambda|<1$ is equal to the number of positive eigenvalues of $\mP$.
            \item The number of eigenvalues $\lambda$ of $\A$ such that $|\lambda|>1$ is equal to the number of negative eigenvalues of $\mP$.
        \end{itemize}
    \end{therm}
    After a change of basis (that we detail in Section~\ref{sec:algorithm} with the approximation algorithm), we can rewrite $\widehat{\A}$ in block-diagonal form: 
    \begin{equation}\label{defplus}
            \widehat{\A}= \begin{pmatrix}
                      \widehat{\A}_+          & \mat{0}                \\
                      \mat{0}                 & \widehat{\A}_-     
                      \end{pmatrix}
    \end{equation}
    where the modulus of the eigenvalues of $\widehat{\A}_+$ (resp. $\widehat{\A}_-$) is smaller (resp. greater) than one. We then apply the same change of coordinates on $\widehat{\balpha}$ and $\widehat{\bbeta}$.
    
    We can finally find the rational component of the function $\psi$, \emph{i.e.} the function $g$ from Corollary~\ref{corollary:stable} necessary to solve that approximate minimization problem.
    \begin{therm}
        Let $\widehat{\A}_+, \widehat{\balpha}_+, \widehat{\bbeta}_+$ be as in Equation \ref{defplus}. The rational component of $\psi$ is the function $g= \widehat{\balpha}_+^{\top}(z\mat{1}-\widehat{\A}_+)^{-1}\widehat{\bbeta}_+$.
    \end{therm}
    \begin{proof}
        Clearly $\psi=g+ l$, with $l=\widehat{\balpha}_-^{\top}(z\mat{1}-\widehat{\A}_-)^{-1}\widehat{\bbeta}_-$, $l \in \mathcal{H}^{\infty}$. To conclude the proof we need to show that $g$ has $k$ poles inside the unit disc, and that therefore it has rank $k$. We do this by studying the modulus of the eigenvalues of $\widehat{\A}_+$. 
        
        Since $E$ is minimal, $\widehat{A}$ is reachable by definition, so we can use Theorem~\ref{theorem:inertia} and solve the problem by directly examining the eigenvalues of $-\bm{\Sigma} \mat{R}$. From the proof of Theorem~\ref{Theorem:maintrm} we have $-\boldsymbol{\Sigma} \mat{R}=\boldsymbol{\Sigma}(\boldsymbol{\Sigma}^2-\sigma^2_k\mat{1})$, where $\boldsymbol\Sigma$ is the diagonal matrix having as elements the singular numbers of $H$ different from $\sigma_k$. It follows that $-\boldsymbol{\Sigma} \mat{R}$ has only $k$ strictly positive eigenvalues, and $\widehat{\A}$ has $k$ eigenvalues with modulus smaller than $1$. Thus, $\widehat{\A}_+$ has $k$ eigenvalues, corresponding to the poles of $g$.
    \end{proof}

\subsubsection{Solving the Approximation Problem}

    Now that we have found the rational function $g$, a symbol for the operator that solves Theorem~\ref{theorem:aakop}, we need to find the parameters of $\widehat{A}_k$, the WFA corresponding to the optimal approximation. These are directly revealed by the expression of $g$, due to the function's parametrization.
    
    \begin{therm}
         Let $A= \langle \balpha, \A, \bbeta \rangle$ be a minimal WFA with $n$ states over a one-letter alphabet. Let $A$ be in its SVA form. The optimal spectral-norm approximation of rank $k$ is given by the WFA $\widehat{A}_k= \langle \widehat{\balpha}_+, \widehat{\A}_+, \widehat{\bbeta}_+ \rangle$.
    \end{therm}
    \begin{proof}
        From Corollary~\ref{corollary:stable} we know that $g$ is the rational function associated with the Hankel matrix of the best approximation. Given the correspondence between the Fourier coefficients of $g$ and the entries of the matrix, we have:
        \begin{equation}
            g= \widehat{\balpha}_+^{\top}(z\mat{1}-\widehat{\A}_+)^{-1}\widehat{\bbeta}_+ = \sum_{k\geq 0} \widehat{\balpha}_+^{\top}\widehat{\A}_+^k \widehat{\bbeta}_+ z^{-k-1} =\sum_{k\geq 0}\bar{f}(k) z^{-k-1}
        \end{equation}
    where $\bar{f}:\Sigma^* \rightarrow \R$ is the function computed by $\widehat{A}_k$ and $\widehat{\balpha}_+, \widehat{\A}_+, \widehat{\bbeta}_+$ are the parameters.
    \end{proof}

\subsection{Error Analysis}

    Thanks to the use of AAK theory, the method outlined in the previous sections is guaranteed to return the rank $k$ optimal spectral-norm approximation of a WFA satisfying our assumptions, and the singular number $\sigma_k$ provides the error. As noticed before, since the Hankel matrix has finite rank and we can derive the Gramian matrices of the WFA, the singular number corresponding to the error can be computed precisely, even though the Hankel matrix is infinite.
    
    Similarly to the case of SVA truncation~\citep{Balle19}, owing to the ordering of the singular numbers, the error decreases when $k$ increases, meaning that allowing $\widehat{A}_k$ to have more states guarantees a better approximation of $A$. Note that  the solution we propose is optimal in the spectral norm, but it might not be the case in other norms. Nonetheless, we have the following bound between $\ell^2$ norm and spectral norm. 
  
    \begin{therm}\label{l2bound}
        Let $A$ be a minimal WFA computing $f:\Sigma^* \rightarrow \R$, with matrix $\H$. Let $\widehat{A}_k$ be its optimal spectral-norm approximation, computing $g:\Sigma^* \rightarrow \R$, with matrix $\mat{G}$. Then:
        \begin{equation}
            \norm{f-g}_{\ell^2} \leq \norm{\H - \mat{G}} = \sigma_k.
        \end{equation}
    \end{therm}
    
    \begin{proof}
        Let $\mat{e}_0=\begin{pmatrix} 1 & 0 & \cdots \end{pmatrix}^{\top}$, $f:\Sigma^*\rightarrow \R$, $g:\Sigma^*\rightarrow \R$ with Hankel matrices $\H$ and $\mat{G}$, respectively. We have:
        \begin{align*}
            \norm{f-g}_{\ell^2} &=\left(\sum_{n=0}^{\infty}|f_n-g_n|^2 \right)^{1/2}\\
            &=\norm{(\H-\mat{G}) \mat{e}_0}_{\ell^2}\\
            &\leq \sup_{\norm{\mat{x}}_{\ell^2}=1}\norm{(\H-\mat{G})\mat{x}}_{\ell^2}\\
            &=\norm{\H-\mat{G}}=\sigma_k 
        \end{align*}
        where the second equation follows by definition and by observing that matrix difference is computed entry-wise. 
    \end{proof}

\section{Algorithm}\label{sec:algorithm}

    We now use the results obtained in the previous sections to define Algorithm~\ref{alg:approx}, that we call \texttt{AAKapproximation}.
    
    The algorithm takes as input a target number of states $k<n$, a minimal irredundant WFA $A$ $n$ states and in SVA form, and its Gramian $\mat{P}$. We assume $\balpha_2 \neq 0$. If $\balpha_2 = 0$, it is enough to substitute the Steps $4,5,6$ with the analogues from Equation~\ref{eq:balphazero}. As mentioned in Section~\ref{assumptions}, the constraints on the WFA $A$ to be minimal and in SVA form are not essential. In fact a WFA with $n$ states can be minimized in time $O(n^3)$~\citep{rationalseries}, and the SVA computed in $O(n^3)$~\citep{Balle19}. The algorithm applies the results of Theorem~\ref{Theorem:maintrm} in order to derive the parameters of the optimal WFA. The output of the algorithm is the WFA $\widehat{A}_k$ corresponding to the unique optimal spectral-norm approximation of $A$.

    \begin{algorithm}[t]
    \caption{\texttt{AAKapproximation}}\label{alg:approx}
        \SetAlgoVlined
        \DontPrintSemicolon
        \SetKwInOut{Input}{input}
        \SetKwInOut{Output}{output}
        \Input{A minimal WFA $A$, with $\balpha_2\neq 0$, $n$ states and in SVA form, \newline its Gramian $\mat{P}$, a target number of states $k<n$}
        \Output{A WFA $\widehat{A}_k$ with $k$ states}
        Let $\balpha_1,\balpha_2,\bbeta_1,\bbeta_2,\A_{11},\A_{12},\A_{22},\bm{\Sigma}$ be the blocks defined in Eq.~\ref{eq:setup}\;
        Let $(\balpha_2^{\top})^+= \frac{\balpha_2}{\balpha_2^{\top}\balpha_2}$, $(\bbeta_2^{\top})^+= \frac{\bbeta_2}{\bbeta_2^{\top}\bbeta_2}$\;
        Let $\mat{R}=\sigma_k^2\mat{1}-\bm{\Sigma}^2$\;
        Let $\widehat{\A}=(\A_{11}^{\top}- \A_{21}^{\top}(\bbeta_2^{\top})^{+}\bbeta_1^{\top})^{-1}$\;
        Let $\widehat{\balpha} = \widehat{\A}^{\top}\mat{R}\A_{12}(\balpha_2^{\top})^{+}$\;
        Let $\widehat{\bbeta} = - \widehat{\A}\A_{21}^{\top}(\bbeta_2^{\top})^{+}$\;
        Let $\widehat{A}=\langle \widehat{\balpha}, \widehat{\A}, \widehat{\bbeta} \rangle$\;
        Let $\widehat{A}_k \leftarrow$ \texttt{BlockDiagonalize}($\widehat{A}$)\;
        \Return $\widehat{A}_k$ 
    \end{algorithm}
    
    \paragraph*{Block Diagonalization}\label{blockdiag}
    The algorithm involves a call to Algorithm~\ref{alg:blockd}, \texttt{BlockDiagonalize}. This algorithm corresponds to the steps
    %, outlined in Section~\ref{blockdiag}, 
    necessary to derive the WFA $\widehat{A}_k$ associated to the rational function $g$. One way to solve the problem is to compute the Jordan form of the matrix. Unfortunately, this problem is ill-conditioned, so it is not suitable for our algorithmic purposes. Following an idea of \citet{Glover}, we compute the Schur decomposition, \emph{i.e.} we find an orthogonal matrix $\mat{U}$ such that the matrix $\mat{U}^{\top}\widehat{\A}\mat{U}$ is upper triangular, with the eigenvalues of $\widehat{\A}$ on the diagonal. We obtain:
    \begin{equation}\label{eq:T}
        \mat{T}=\mat{U}^{\top}\widehat{\A}\mat{U}= \begin{pmatrix}
                      \widehat{\A}_{+}     & \widehat{\A}_{12}               \\
                      \mat{0}                  & \widehat{\A}_{-}     
                      \end{pmatrix}
    \end{equation}
    where the eigenvalues are arranged in increasing order of modulus, and the modulus of those in $\widehat{\A}_{+}$ (resp. $\widehat{\A}_{-}$) is smaller (resp. greater) than one. To transform this upper triangular matrix into a block-diagonal one, we use the following result.
    \begin{therm}[\cite{Roth}]\label{bart}
        Let $\mT$ be the matrix defined in Equation~\ref{eq:T}. The matrix $\mat{X}$ is a solution of the equation $\widehat{\A}_{+}\mat{X}- \mat{X}\widehat{\A}_{-} +\widehat{\A}_{12}=\mat{0}$ if and only if the matrices
        \begin{equation}
            \mat{M}=\begin{pmatrix}
                      \mat{1}     & \mat{X}               \\
                      \mat{0}                  & \mat{1}     
                      \end{pmatrix}, \quad \text{and} \quad \mat{M}^{-1}=\begin{pmatrix}
                      \mat{1}     & -\mat{X}               \\
                      \mat{0}                  & \mat{1} \end{pmatrix}
        \end{equation}
        satisfy:
    \begin{equation}
        \mat{M}^{-1}\mat{T}\mat{M}=\begin{pmatrix}
                      \widehat{\A}_{+}     & \mat{0}               \\
                      \mat{0}                  & \widehat{\A}_{-}     
                      \end{pmatrix},
    \end{equation}
    where $\mT$ is the matrix defined in Equation~\ref{eq:T}.
    \end{therm}
    Setting $\boldsymbol{\Gamma}=\begin{pmatrix}\mat{1}_k & \mat{0} \end{pmatrix}$ we can now derive the rational component of the WFA:
    \begin{align}
        &\widehat{\A}_+= \boldsymbol{\Gamma}\mat{M}^{-1}\mat{U}^{\top}\widehat{\A}\mat{U}\mat{M}\boldsymbol{\Gamma}^{\top}\\
        &\widehat{\balpha}_+= \boldsymbol{\Gamma} \mat{M}^{\top}\mat{U}^{\top}\widehat{\balpha}\\    
        &\widehat{\bbeta}_+= \boldsymbol{\Gamma}\mat{M}^{-1}\mat{U}^{\top}\widehat{\bbeta}.
    \end{align}
    The algorithm \texttt{BlockDiagonalize} corresponds to the implementation of this procedure, and Step $2$ can be performed using the Bartels-Stewart algorithm~\citep{BartelStew}.

    \begin{algorithm}[t]
    \caption{\texttt{BlockDiagonalize}}\label{alg:blockd}
        \SetAlgoVlined
        \DontPrintSemicolon
        \SetKwInOut{Input}{input}
        \SetKwInOut{Output}{output}
        \Input{A WFA $\widehat{A}$}
        \Output{A WFA $\widehat{A}_k$ wit $\rho<1$}
        \If{$\dim{\widehat{\A}}=k$}
        {\Return $\widehat{A}_k$}
        \Else
         {Compute the Schur decomposition of $\widehat{\A}=\mat{U}\mat{T}\mat{U}^{\top}$, where $|T_{11}|\leq |T_{22}| \leq \dots$\;
        Solve $\widehat{\A}_{11}\mat{X}- \mat{X}\widehat{\A}_{22}+\widehat{\A}_{12}=\mat{0}$ for $\mat{X}$\;
        Let $\mat{M}=\begin{pmatrix}
                      \mat{1}     & \mat{X}               \\
                      \mat{0}                  & \mat{1}     
                      \end{pmatrix}$ and $\mat{M}^{-1}=\begin{pmatrix}
                      \mat{1}     & -\mat{X}               \\
                      \mat{0}                  & \mat{1} \end{pmatrix}$\;
        Let $\boldsymbol{\Gamma}=\begin{pmatrix}\mat{1}_k & \mat{0} \end{pmatrix}$\;
        Let $\widehat{\A}_+= \boldsymbol{\Gamma}\mat{M}^{-1}\mat{U}^{\top}\widehat{\A}\mat{U}\mat{M}\boldsymbol{\Gamma}^{\top}$\;
        Let $\widehat{\balpha}_+= \boldsymbol{\Gamma} \mat{M}^{\top}\mat{U}^{\top}\widehat{\balpha}$\;
        Let $\widehat{\bbeta}_+= \boldsymbol{\Gamma}\mat{M}^{-1}\mat{U}^{\top}\widehat{\bbeta}$\;
        Let $\widehat{A}_k=\langle \widehat{\balpha}_+, \widehat{\A}_+, \widehat{\bbeta}_+ \rangle$\;
        \Return $\widehat{A}_k$}
    \end{algorithm}

\subsection{Computational Cost}
    
    The running time of \texttt{BlockDiagonalize} with input a WFA $\widehat{A}$ with $(n-r)$ states is thus in $O((n-r)^3)$, where $r$ is the multiplicity of the singular value considered. The running time of \texttt{AAKapproximation} for an input WFA $\widehat{A}$ with $n$ states is in $O((n-r)^3)$. In particular, it is possible to analyze the cost associated to each step of the algorithms~\citep{Computationalcost}:
    
    \begin{itemize}
        \item The product of two $n\times n$ matrices can be computed in time $O(n^3)$ using a standard iterative algorithm.
        %, but can be reduced to $O(n^{\omega})$ with $\omega<2.4$.
        \item The inversion of a $n\times n$ matrix can be computed in time $O(n^3)$ using Gauss-Jordan elimination.
        %, but can be reduced to $O(n^{\omega})$ with $\omega<2.4$.
        \item The computation of the Schur decomposition of a $n\times n$ matrix can be done with a two-step algorithm, where each step takes $O(n^3)$, using the Hessenberg form of the matrix. 
        \item The Bartels-Stewart algorithm applied to upper triangular matrices to find a matrix of size $m\times n$ takes $O(mn^2+nm^2)$.
    \end{itemize}

\section{Example} \label{sec:example}

    We consider the following weighted finite automaton with three states over a one-letter alphabet, represented in SVA form:
         \begin{equation*}
            \A= \begin{pmatrix}
                    0.579 & 0.461 & 0.046    \\
                    -0.461 & -0.192 & 0.225   \\
                    0.046 & -0.225 & -0.387
                \end{pmatrix} , \quad
            \balpha= \begin{pmatrix}
                1.650     \\
                -0.851 \\
                0.038
                \end{pmatrix}, \quad
            \bbeta= \begin{pmatrix}
                1.650     \\
                0.851 \\
                0.038
                \end{pmatrix},
        \end{equation*}
    
    The objective is to find the WFA with two states solving the approximate minimization problem optimally.

    We first note that $\A$ has spectral radius strictly smaller than $1$, having eigenvalues:
        \begin{equation}
            \lambda_{1,2}=0.0162324 \pm 0.0297233 i \quad\quad \lambda_3=0.0324648.
        \end{equation}
     Therefore, the assumptions listed Section \ref{assumptions} are satisfied, and we can apply Theorem \ref{Theorem:maintrm}. We compute the gramian matrices and obtain, according to the partition in Equation \ref{eq:setup}, the following matrix: 
        \begin{equation*}
            \mP= \mQ=\begin{pmatrix}
                    4.67 & 0 & 0    \\
                    0 & 1.79 & 0   \\
                    0 & 0 & 0.12
                \end{pmatrix},
        \end{equation*}
        so that $\sigma_2^2=0.12$ and:
        \begin{equation*}
            \bm{\Sigma}=\begin{pmatrix}4.67 & 0    \\
                    0 & 1.79
            \end{pmatrix}.
        \end{equation*}
         We then proceed by partitioning $\A$, $\balpha$ and $\bbeta$ and obtain:
        \begin{equation*}
            \A_{11}= \begin{pmatrix}
                    0.579 & 0.461    \\
                    -0.461 & -0.192
                \end{pmatrix}, \quad \A_{1,2}= \begin{pmatrix} 0.046    \\
                0.225   
                \end{pmatrix}, \quad \A_{2,1}^{\top}= \begin{pmatrix} 0.046    \\
                -0.225   
                \end{pmatrix}, \quad \A_{22}=-0.387.
        \end{equation*}
        \begin{equation*}
            \balpha_1= \begin{pmatrix}
                1.650     \\
                -0.851
                \end{pmatrix}, \quad \bbeta_1= \begin{pmatrix}
                1.650     \\
                0.851
                \end{pmatrix}, \quad \balpha_2=\bbeta_2=0.038.
        \end{equation*}

        Since $\balpha_2\neq0$, we can use Equation \ref{eq:icalpsol} to find the coefficients of the auxiliary WFA $\widehat{A}=\langle\widehat{\balpha},\widehat{\A},\widehat{\bbeta}\rangle$.
        
        We have:
        \begin{align*}
            &\begin{cases}
                \widehat{\bbeta} = - \widehat{\A}\A_{21}^{\top}(\bbeta_2^{\top})^{+} \\
                \widehat{\balpha} = \widehat{\A}^{\top}\mat{R}\A_{12}(\balpha_2^{\top})^{+}\\
                \widehat{\A}(\A_{11}^{\top}- \A_{21}^{\top}(\bbeta_2^{\top})^{+}\bbeta_1^{\top})=\mat{1}
            \end{cases}\\
            &\begin{cases}
                \widehat{\bbeta} = - \widehat{\A}\begin{pmatrix} 0.046    \\
                -0.225   
                \end{pmatrix}(0.038)^{-1} \\
                \widehat{\balpha} = \widehat{\A}^{\top}\left(\begin{pmatrix}0.12 & 0    \\
                    0 & 0.12
            \end{pmatrix}-\begin{pmatrix}4.67 & 0    \\
                    0 & 1.79
            \end{pmatrix}^2\right)\begin{pmatrix} 0.046    \\
                0.225   
                \end{pmatrix}(0.038)^{-1}\\
                \widehat{\A}\left(\begin{pmatrix}
                    0.579 & 0.461    \\
                    -0.461 & -0.192
                \end{pmatrix}^{\top}- \begin{pmatrix} 0.046    \\
                -0.225   
                \end{pmatrix}(0.038)^{-1}\begin{pmatrix}
                1.650     \\
                0.851
                \end{pmatrix}^{\top}\right)=\mat{1}
            \end{cases}
        \end{align*}
        so we get:
        \begin{equation*}
            \widehat{\A}= \begin{pmatrix} 0.578 & 0.178 \\ -1.221 & -0.169 
            \end{pmatrix} , \quad
            \widehat{\balpha}=  \begin{pmatrix} 7.105    \\
                 -1.579
                \end{pmatrix}, \quad
            \widehat{\bbeta}= \begin{pmatrix} 0.353    \\
                 0.474
                \end{pmatrix}.
        \end{equation*}
        
        Now, we want to extract the rational component in order to find the optimal approximation. To do so, we block-diagonalize the transition matrix $\widehat{\A}$ and look at the modulus of its eigenvalues. We have:
        \begin{equation*}
            \lambda_{1,2}=0.204593 \pm 0.278322 i.
        \end{equation*}
        As we can see, both eigenvalues have modulus smaller than one. This means that the WFA $\widehat{A}$ is exactly the optimal approximation of size two that we are looking for, and there aren't any components that need to be discarded. Following the notation introduced in the previous section, we have: $\widehat{A}_k= \langle \widehat{\balpha}_+, \widehat{\A}_+, \widehat{\bbeta}_+ \rangle= \langle\widehat{\balpha},\widehat{\A},\widehat{\bbeta}\rangle$.

\section{Related Work}\label{relatedwork}

    The problem of minimizing automata has been an important subject of research since the 1950s. There is a remarkable algorithm due to Brzozowski \citep{Brzozowski62,Brzozowski64} that reduces a DFA to a minimal one. However, its worst-case running time is exponential in the number of states. Despite this shortcoming, this algorithm has seen a resurgence recently, mainly because it can be generalized to new models, such as weighted automata \citep{droste}. This line of algorithms is based on a new understanding of Brzozowski's algorithm from the point of view of duality \citep{Bonchi12,Bonchi12b,Bonchi14,Bezhanishvili12} and extend readily to other settings. In the context of quantitative systems, like weighted or probabilistic automata, it becomes meaningful to investigate the approximate minimization problem. The study of this problem and of its applications are fairly recent, and only a few works have been published on the subject. A problem analogous to approximate minimization is addressed by Kulesza, Jiang, and Singh for the spectral algorithm. The authors provide a bound on the loss of the learned low-rank model in terms of the singular values that are discarded during training \citep{kulesza2015}. In a previous work, the same group of authors connected spectral learning to the approximation problem of a small class of Hidden Markov models, bounding the error in terms of the total variation distance~\citep{kulesza14}. Still in the context of Hidden Markov models, Kotsalis and Shamma provide bounds for the model reduction problem using the spectral norm as a measure of the error \citep{HMMSHankel}. We remark that the framework of Hidden Markov models is encompassed by weighted automata \citep{denishmm}. Balle, Panangaden, and Precup are the first authors to formalize the approximate minimization problem for WFAs~\citep{Balle15,Balle19}. The technique presented in their paper relies on the construction (and truncation) of the singular value automaton, a canonical expression for WFAs arising from the singular value decomposition of the corresponding Hankel matrix. Their method can be viewed as a generalization to multi-letter alphabets of the balanced realization approach from control theory \citep{antoulas}. The authors conclude their analysis by providing bounds on the approximation error in the $\ell^2$ norm. The result is supported by strong theoretical guarantees and applies to a large class of WFAs. This method has later been extended to the setting of weighted tree automata in~\citet{ballerabusseau}. The main limitation of these approaches based on SVA truncation is that the approximation obtained is not optimal in any norm. We partially address this point in this work, where we obtain an algorithm for the optimal approximation in the spectral norm for the same class of WFAs considered by Balle, Panangaden, and Precup, but restricted to a one-letter alphabet. Part of this results where presented in \citep{AAK-WFA}. In \cite{AAK-RNN} we extend this results to the more general setting of black-box models trained for language modelling over one-letter alphabets. In \cite{AAK-Learnaut,mythesis} we analyze the problem of extending the method presented in this paper to the case of multi-letter alphabets.
    
    The control theory community has largely studied approximate minimization in the context of linear time-invariant systems~\citep{antoulas}. A parallel with these results can be drawn by noting that the impulse response of a discrete Single-Input-Single-Output SISO system can be parametrized as a WFA over a one-letter alphabet. \citet{Glover} presents a state-space solution for the case of continuous Multi-Input-Multi-Output MIMO systems. His method led to a widespread application of these results, thanks to its computational and theoretical simplicity. This stems from the structure of the continuous Lyapunov equations. For discrete systems, though, the quadratic nature of the Lyapunov equations does not allow for a simple closed form formula for the state space solution~\citep{discreteH}. Thus, most of the results for the discrete case work with a suboptimal version of the problem~\citep{balldiscrete, Al-Hussari, Ionescu}. A solution for the SISO case can be found using a polynomial approach, but it does not provide an explicit representation of the state space nor it generalizes to the MIMO setting. The first to actually extend Glover results is Gu, who provides an elegant solution for the MIMO discrete problem~\citep{gu}. Glover and Gu's solutions rely on building an
    %embedding the initial system into an extension of it,
    \emph{all-pass system}, equivalent to the WFA $E$ in our case. Part of our contribution is the adaptation of some of the control theory tools to WFAs.

\section{Extensions and Future Work}\label{future_work}
    
    In this section we examine possible extensions of our method by relaxing some of the hypothesis.

\subsection{Removing the Finite-Rank Assumption}

    The proof of Theorem~\ref{theorem:aakop} is constructive for any compact Hankel operator. In the setting of this paper, compactness is guaranteed, as the operator corresponding to an irredundant WFA has finite rank and is bounded. While boundness is necessary for compactness, the finite-rank hypothesis is not. Therefore, an interesting extension of this work is to investigate other classes of models by relaxing the finite-rank (or finite state) assumption. An example of models corresponding to infinite-rank Hankel matrices are Recurrent Neural Networks (RNNs)~\citep{Schmidhuber}. Recently, particular attention has been given to the problem of extracting, from an RNN, a weighted finite automaton~\citep{Ayache2018,Rabusseau19,WeissWFA19,Takamasa,eyraud2020,probmod,zhangaaai}. In this sense, the knowledge distillation task~\citep{hinton2015distilling} is very similar to an approximate minimization problem, since WFAs are a less expensive alternative to RNNs, while still being expressive and suited for sequence modelling and prediction~\citep{denishmm,cortes}. In \citet{AAK-RNN}, we investigated the use of AAK theory on black-box models trained for language modelling on sequential data. In particular, we showed that compactness is automatically respected by black boxes for language modelling, and proposed an algorithm for the one-letter setting, based on AAK theory. This particular extension of the method presented in this paper constitutes a first fundamental step towards developing provable approximation algorithms for black box models.

\subsection{Removing the Spectral Radius Assumption}

    One could consider a WFA over a one-letter alphabet with $\rho(\A)\neq 1$, \emph{i.e.} not necessarily irredundant. In this case, the method proposed in the previous sections can be extended and the quality of the approximation can be estimated, but the result is not optimal in the spectral norm. %Nonetheless, we can evaluate the quality of the approximation using a bound on the error.
    Once again, we draw inspiration from the control theory literature, where some theoretical work has been done to study an analogous approach for continuous time systems and their approximation error~\citep{Glover}.
    
    The key idea is to block-diagonalize $\A$ like we did in Section~\ref{sec:rational_comp}. This way, we obtain two components, $\A_+$ and $\A_-$, with the property that $\rho<1$ and $\rho>1$, respectively. We tackle each component separately. The case of $A_+=\langle \balpha_+, \A_+, \bbeta_+ \rangle$, the component having $\rho(\A)<1$, can be dealt with in the way presented in the previous sections. This means that we can find an optimal spectral-norm approximation of the desired size for $A_+$. Then, we can consider the second component, $A_-=\langle \balpha_-, \A_-, \bbeta_- \rangle$. In this case, we apply the transformation
    \begin{equation*}
        z^{j-1} \mapsto z^{-j} \quad\text{for}\,\, j\geq 1
    \end{equation*}
    to the symbol $\phi'(z)$ associated to $A_-$. Then, the function
    \begin{equation*}
        \phi'(z^{-1})= \sum_{k\geq 0} \balpha_-^{\top}\A_-^k z^k \bbeta_- = \balpha_-^{\top}(\mat{1}-z\A_-)^{-1}\bbeta_-
    \end{equation*}
    is well defined, as the series converges for $z$ with small enough modulus. The use of this transformation allows us to obtain a function having poles only inside the unit disc, and to apply the method presented in this chapter. We remark that in this case an important choice to make is the size of the target approximation of $A_-$, as it can influence the quality of the result. Analyzing the effects of this parameter on the approximation error is an interesting direction for future work, both on the theoretical and experimental side.

\subsection{Removing the One-Letter Assumption}

    The most pressing direction for future work is undoubtedly to extend our results to a multi-letter setting. The work of Adamyan, Arov and Krein provides us with a powerful theory connecting sequences to the study of complex functions. Unfortunately, this approach cannot be directly generalized to the multi-letter case, when $\Sigma^*$ is a noncommutative monoid, as it requires to generalize standard harmonic analysis results to the non-abelian case. A recent line of work in multivariable operator theory has been centered around extending results of standard operator theory to the case of noncommutative operators defined on Fock spaces \cite{frazho,bunce,popescu_arias,popescu1,popescu_old,popescu3,popescuthesis,popescu_oponfock,popescu,popescu_free,popescu_domain,popescu2,ball_bolotnikov_2021,NCrational}. In particular, a noncommutative definition of Hankel operator, and a noncommutative version of the AAK theorem are presented in a recent work of Popescu \cite{popescu}, but its proof is not constructive. Therefore, solving the approximate minimization problem for multi-letter alphabets using AAK theory comes with two distinct challenges:
    \begin{itemize}
        \item \emph{Finding a noncommutative Hankel operator:} given a WFA and its Hankel matrix, we need to find a way to reformulate the approximation problem using multivariable operators. In particular, we need to find a noncommutative analogue of the Hardy space and of the symbol.
        \item \emph{Making AAK constructive:} the proof of the noncommutative version of the AAK theorem does not provide us with an expression for the optimal approximation. An interesting direction would be to explore ways to extend the proof to a constructive one.
    \end{itemize}
    
    In \citet{AAK-Learnaut}, we proposed a framework to associate a noncommutative Hankel operator (defined on an noncommutative version of the Hardy space) and a noncommutative rational function to the Hankel matrix computed by a model on sequential data, solving the first point listed above. In the one-letter case, obtaining the framework allowed us to reformulate the approximation problem in terms of functional analysis, and to solve it using the constructive proof of AAK theorem. In \citet{mythesis}, we tried to address the question of whether or not the proof of the noncommutative AAK theorem can be made constructive. While we did not manage to provide a definitive answer, we laid out possible approaches that can be used to tackle the problem of making the proof of the noncommutative version of AAK theorem constructive.

\section{Conclusion}\label{conclusion}

    In this paper we applied the AAK theory for Hankel operators and complex functions to the framework of WFAs in order to construct the optimal approximation to an automaton given a bound on the size. We propose an algorithm to find the parameters of the best WFA approximation in the spectral norm, and derive bounds on the error. Our method applies to real irredundant WFAs defined over a one-letter alphabet. These alphabets have proven to be of independent interest when dealing with automata, as in this case the classes of regular and context-free languages collapse~\citep{Pighizzini}. 
    
    We think the spectral norm has desirable characteristics, making it a solid candidate for the approximate minimization task. For example, it can be minimized in polynomial time and a global minimum for the error can be computed accurately. Moreover, the fact that this norm is independent on the specific architecture or model considered, facilitate future applications of this method, as it can be used to compare different classes of models. Nonetheless, a limitation of this work is that we do not have a clear picture of how effective it is to use the spectral norm to evaluate the approximation of WFAs and black boxes. Concretely, we do not know how the spectral norm performs with respect to behavioral metrics, or other metrics coming from natural language processing~(e.g., word error rate and normalized discounted cumulative gain). To some extent, this problem is a collateral effect of the size of the alphabet: the comparison between spectral norm and other kind of norms is possible only in the multi-letter setting. Obtaining algorithms for the multi-letter case will thus open the possibility of evaluating the quality of the spectral norm. 

    While the one-letter setting is certainly restricted, we believe that this work constitutes a first fundamental step in the direction of optimal approximation.
    Furthermore, the use of AAK techniques has proven to be very fruitful in related areas like control theory; we think that automata theory can also benefit from it. The use of such methods can help deepen the understanding of the behaviour of rational functions. This paper highlights and strengthens the interesting connections between functional analysis, automata theory and control theory, unifying tools from different domains in one formalism.

\section*{Acknowledgments}
    This research has been supported by NSERC Canada (C. Lacroce, P. Panangaden) and Canada CIFAR AI chairs program (G. Rabusseau). The authors would like to thank Tianyu Li, Harsh Satija and Alessandro Sordoni for feedback on earlier drafts of this work, Gheorghe Comanici and Robert Robere for a detailed review, Florence Clerc for help with the submission, and Maxime Wabartha for fruitful discussions and comments on proofs.
    
\newpage

\bibliography{bibliography}
\newpage
\appendix

\section{Technical Results}\label{apd:first}

\paragraph*{Singular functions}
    We show how to compute the functions corresponding to  a $\sigma_k$-Schmidt pair, and that their quotient is indeed unimodular.
   \begin{therm}\label{theorem:singval}
        Let $\sigma_k$ be a singular number of the Hankel operator $H$. The singular functions associated with the $\sigma_k$-Schmidt pair $\{\boldsymbol{\xi}_k, \boldsymbol{\eta}_k\}$ of $H$ are:
        \begin{align}
            \xi^{+}_k(z)&=\sigma_k^{-1/2}\bbeta^{\top}(\mat{1} - z\A)^{-1}\mat{e}_k\\
            \eta^{-}_k(z)&=\sigma_k^{-1/2}\balpha^{\top}(z\mat{1}-\A^{\top})^{-1}\mat{e}_k.
        \end{align}
        If $\psi$ is the best approximation to the symbol, then $\sigma_k^{-1}e$ has modulus $1$ almost everywhere on the unit circle (\emph{i.e.} it is unimodular).
    \end{therm}
    
    \begin{proof}
    
        Let $\mat{F}$ and $\mat{B}$ be the forward and backward matrices, respectively, with $\H=\mat{F}\mat{B}^{\top}$, $\mP=\mat{F}^\top \mat{F}, \mQ=\mat{B}^\top \mat{B}$. We consider the $\sigma_k$-Schmidt pair $\{\boldsymbol{\xi}_k, \boldsymbol{\eta}_k\}$. By definition, $\H^{\top}\H\boldsymbol{\xi}_k=\sigma_k^2\boldsymbol{\xi}_k$. By rewriting in terms of the FB factorization, we obtain:
        \begin{align}
            &\H^{\top}\H\boldsymbol{\xi}_k=\sigma_k^2\boldsymbol{\xi}_k\\
            &\mat{B}\mat{F}^{\top}\mat{F}\mat{B}^{\top}\boldsymbol{\xi}_k=\sigma_k^2\boldsymbol{\xi}_k\\
            &\mat{B}\mP\mat{B}^{\top}\boldsymbol{\xi}_k=\sigma_k^2\boldsymbol{\xi}_k\\
            &\mat{B}\mP\mat{e}_k=\sigma_k^2\boldsymbol{\xi}_k
        \end{align} 
        where in the last step we set $\mat{e}_k= \mat{B}^ \top \boldsymbol{\xi}_k$, to reduce the SVD problem of $\H$ to the one of $\mQ\mP$. Note that, since $\mP$ and $\mQ$ are diagonal, $\mat{e}_k$ is the $k$-th coordinate vector $(0,\dots,0,1,0,\dots,0)^{\top}$. Since $\mat{e}_k$ is an eigenvector of $\mQ\mP$ for $\sigma_k^2$, we get:
        \begin{align}
            &\mat{B}\mQ^{-1}\mQ\mP\mat{e}_k=\sigma_k^2\boldsymbol{\xi}_k\\
            &\mat{B} \mQ^{-1} \mat{e}_k= \boldsymbol{\xi}_k.
        \end{align}
        Moreover, $\H$ is symmetric, so we have that the singular vectors $\boldsymbol{\eta}_k$ and $\boldsymbol{\xi}_k$ have the same coordinates up to the sign of the corresponding eigenvalues.
        We obtain:
        \begin{align}
            \xi^+_k(z)&=\sum_{j=0}^{\infty} \sigma_k^{-1/2}\bbeta^{\top}\A^j \mat{e}_k z^{j}= \sigma_k^{-1/2}\bbeta^{\top}(\mat{1} - z\A)^{-1}\boldsymbol{e}_k\\
            \eta^{-}_k(z)&=\sum_{j=0}^{\infty}\sigma_k^{-1/2}\balpha^{\top}\A^{j\top} \mat{e}_k z^{-j-1}=\sigma_k^{-1/2}\balpha^{\top}(z\mat{1}-\A^\top)^{-1}\boldsymbol{e}_k
        \end{align}
        where the singular functions have been computed following Equation~\ref{eq:hardynotation}. If $r$ is the multiplicity of $\sigma_k$, from Corollary~\ref{corollary:unimodular} we get the following fundamental equation:
        \begin{equation*}
            (\phi-\psi)\bbeta^{\top}(\mat{1} - z\A)^{-1}\mat{V}= \sigma_k \balpha^{\top}(z\mat{1}-\A^\top)^{-1}\mat{V}
        \end{equation*}
        where $\mat{V}=\begin{pmatrix}
            \mat{0} & \mat{1}_r
          \end{pmatrix}^\top$ is a $n\times r$ matrix.
        Consequently, we obtain the function:
        \begin{equation*}
            \sigma_k^{-1}e= \frac{\balpha^{\top}(z\mat{1}-\A^\top)^{-1}\mat{V}}{\bbeta^{\top}(\mat{1} - z\A)^{-1}\mat{V}}
        \end{equation*}
        which is unimodular, since $\balpha_i=\operatorname{sgn}(\lambda_i)\bbeta_i$, and $\A=\operatorname{sgn}(\lambda_i)\A^\top$.
    \end{proof}  
    
     \paragraph*{Proof of Theorem~\ref{Theorem:allpass}.} In order to prove Theorem~\ref{Theorem:allpass} we need an auxiliary lemma \cite[Lemma 6.1]{discreteH}. These are the analogous of some control theory results, rephrased in terms of WFAs. The original theorem and lemma, together with the corresponding proofs, can be found in~\cite{discreteH}. Hence, we only provide a sketch of the proofs.
    
    \begin{lemma}[\cite{discreteH}]\label{lemma:T}
        Let $E=\langle \balpha_e, \A_e, \bbeta_e \rangle$ be a minimal WFA. Let $e(z)= \balpha_e^{\top}(z\mat{1}-\A_e)^{-1} \bbeta_e - C$, if $\sigma_k^{-1}e(z)$ is unimodular, then there exist a unique invertible symmetric matrix $\mT$ satisfying:
        \begin{itemize}
            \item[(a)] $\A_e^{\top}\mT\bbeta_e=\balpha_e C$
            \item[(b)] $\sigma_k^2\balpha_e^{\top}\mT^{-1}\A_e^{\top}=C\bbeta_e^{\top}$
            \item[(c)] $\A_e^{\top}\mT\A_e-C^{-1}\A_e^{\top}\mT\bbeta_e\balpha_e^{\top}=\mT$
        \end{itemize}
    \end{lemma}
    \begin{proof}
    
        Since $\sigma_k^{-1}e(z)$ is unimodular, we have that:
        \begin{equation}
            e(z)e^*(\bar{z}^{-1})=\sigma_k^2\mat{1}    
        \end{equation}
        where we denote with $e^*$ the adjoint function. From the equation above, we obtain:
        \begin{align}
            e^*(\bar{z}^{-1})&=\sigma_k^2e^{-1}(z)=\sigma_k^2(-C+\balpha_e^{\top}(z\mat{1}-\A_e)^{-1}\bbeta_e)^{-1}\\
                             &=-\sigma_k^2C^{-1}-\sigma_k^2 C^{-1}\balpha_e^{\top}((z\mat{1}-(\A_e+C^{-1}\bbeta_e\balpha_e))^{-1}\bbeta_eC^{-1}\label{eq:inv}
        \end{align}
        where we used the matrix inversion lemma. On the other hand we have:
        \begin{align}
             e^*(\bar{z}^{-1})&=-C+\bbeta_e^{\top}(z^{-1}\mat{1}-\A_e^{\top})^{-1}\balpha_e\\
                             &= -C+\bbeta_e^{\top}(-\A_e^{-\top}(\mat{1}-z\A_e^{\top})+\A_e^{-\top})(\mat{1}-z\A_e^{\top})^{-1}\balpha_e\\                             &=-(C-\bbeta_e^{\top}\A_e^{-\top}\balpha_e)- \bbeta_e^{\top}\A_e^{-\top}(z\mat{1}-\A_e^{-\top})^{-1}\A_e^{-\top}\balpha_e \label{eq:conj}
        \end{align}
        where we used again the matrix inversion lemma before grouping the terms. If the quantities in Equation~\ref{eq:inv} and Equation~\ref{eq:conj} have to be equal, we need their constant term to be the same. Then, we want the $\mathcal{H}^{\infty}_-$-components to correspond, so we consider the corresponding Hankel matrices. It is easy to see that we can once again associate the coefficients of these complex functions to the parameters of a WFA. From the minimality of $E$ we obtain:
        \begin{equation}
            \begin{cases}
                \sigma_k^2C^{-1}\balpha_e^{\top} =\bbeta_e^{\top}\A_e^{-\top}\mT \\
                \A_e+C^{-1}\bbeta_e\balpha_e= \mT^{-1}\A_e^{-\top}\mT\\
                \bbeta_eC^{-1}=\mT^{-1}\A_e^{-\top}\balpha_e
            \end{cases}
        \end{equation}
        where $\mT$ is an invertible matrix~\citep{BalleCLQ14}. This system is equivalent to:
        \begin{equation}\label{eq:systemlemma}
            \begin{cases}
                \sigma_k^2\balpha_e^{\top}\mT^{-1}\A_e^{\top}=C\bbeta_e^{\top}\\
                \A_e^{\top}\mT\A_e-C^{-1}\A_e^{\top}\mT\bbeta_e\balpha_e^{\top}=\mT\\
                \A_e^{\top}\mT\bbeta_e=\balpha_e C
            \end{cases}
        \end{equation}
        To conclude the proof it remains to check that $\mT$ is symmetric, and this can be done by direct computations.
    \end{proof}
    
    \begin{proof}[\textbf{Proof of Theorem~\ref{Theorem:allpass}}]
        This proof follows easily from Lemma~\ref{lemma:T} by setting $\mP=-\sigma^2_k\mT^{-1}$ and $\mQ=-\mT$. We obtain point $(c)$ by direct multiplication. Then, we substitute the last equation in~\ref{eq:systemlemma} into the second one, and we obtain:
        \begin{equation}
            \A_e^{\top}\mT\A_e-\balpha_e\balpha_e^{\top}=\mT    
        \end{equation}
        which verifies point $(b)$ with $\mQ=-\mT$. Point $(a)$ can be obtained analogously combining the first and second equations in~\ref{eq:systemlemma}.
    \end{proof}

\end{document}